\documentclass[12pt]{article}
\pdfoutput=1
\usepackage{framed,enumitem,multirow,float}
\usepackage{amssymb,amsmath,amsthm}
\usepackage{bm}
\usepackage{blkarray}
\usepackage{graphicx}
\usepackage{booktabs}
\usepackage{color}
\usepackage{datetime}
\usepackage[
      colorlinks=true,
      linkcolor=darkblue,  
      urlcolor=blue,    
      filecolor=blue,     
      citecolor=red,
      linktocpage=true,
      pdfstartview=FitV,
      bookmarksopen=true    
      ]{hyperref}
      
\definecolor{darkblue}{rgb}{0.2, 0, 0.8}

\numberwithin{equation}{section}

\newcommand{\dlog}{\text{dlog}\,}
\renewcommand{\th}{^\text{th}}

\renewcommand{\a}{\alpha}
\renewcommand{\b}{\beta}
\renewcommand{\u}{\mathfrak{u}}

\newcommand{\w}{\mathbf{w}}
\newcommand{\tw}{\tilde{\mathbf{w}}}
\newcommand{\s}{\sigma}
\newcommand{\pa}{\partial}
\newcommand{\ts}{\tilde{\s}}
\renewcommand{\r}{\rho}
\newcommand{\tr}{\tilde{\r}}
\newcommand{\N}{\mathcal{N}}
\newcommand{\A}{\mathcal{A}}
\newcommand{\Z}{\mathcal{Z}}
\newcommand{\C}{\mathbb{C}}
\newcommand{\MHV}{\text{MHV}}
\newcommand{\Gr}{\text{Gr}}
\renewcommand{\mod}{\text{mod }}
\renewcommand{\th}{{}^\text{th}}
\newcommand{\beq}{\begin{equation}}
\newcommand{\eeq}{\end{equation}}
\newcommand{\spl}[1]{\begin{split} #1 \end{split}}
\newcommand{\al}[1]{\begin{align} #1 \end{align} }
\newcommand{\als}[1]{\begin{align} \begin{split} #1 \end{split} \end{align}}

\newcounter{algcounter}
\newenvironment{alg}[1][Algorithm \thealgcounter]
{\refstepcounter{algcounter}%
\bgroup\begin{framed}%
\noindent\textbf{#1}\vspace{2mm}\\}
{\end{framed} \egroup}

\newtheorem{lemma}{Lemma}

\makeatletter
\def\namedlabel#1#2{\begingroup
   \def\@currentlabel{#2}%
   \label{#1}\endgroup
}
\makeatother

\begin{document}  


\begin{titlepage}

\begin{flushright}
{\tt MCTP-14-40}
\end{flushright}

\vspace*{1cm}

\begin{center}

{\Large \bf Orientations of BCFW Charts on the Grassmannian}

\vspace*{1cm}

{Timothy M.~Olson}

\vspace{5mm}{\it
Michigan Center for Theoretical Physics \\
Randall Laboratory,
Department of Physics,\\
University of Michigan, 
450 Church St, 
Ann Arbor, MI 48109, USA}\\[4mm]

\bigskip
\texttt{timolson@umich.edu} \\
\end{center}
\bigskip

\begin{abstract}  
The Grassmannian formulation of $\mathcal{N}=4$ super Yang-Mills theory expresses tree-level scattering amplitudes as linear combinations of residues from certain contour integrals. BCFW bridge decompositions using adjacent transpositions simplify the evaluation of individual residues, but orientation information is lost in the process. We present a straightforward algorithm to compute relative orientations between the resulting coordinate charts, and we show how to generalize the technique for charts corresponding to sequences of any not-necessarily-adjacent transpositions. As applications of these results, we demonstrate the existence of a signed boundary operator that manifestly squares to zero and prove via our algorithm that any residues appearing in the tree amplitude sum are decorated with appropriate signs so all non-local poles cancel exactly, not just mod 2 as in previous works.
\end{abstract}

\end{titlepage}

\setcounter{tocdepth}{2}
{\small
\setlength\parskip{-0.5mm}
\tableofcontents
}
\newpage

\section{Introduction}
\namedlabel{sec:intro}{Introduction}
Scattering amplitudes in 4d planar ${\N=4}$ super Yang-Mills (SYM) theory can be formulated as contour integrals in the space of ${k\times n}$ matrices, modulo multiplication by a ${GL(k)}$ matrix; this is the Grassmannian manifold ${\Gr(k,n)}$ \cite{ArkaniHamed:2012nw}. Reformulating scattering amplitudes in this new framework has led to many interesting and unexpected mathematical structures in ${\N=4}$ SYM such as on-shell diagrams \cite{ArkaniHamed:2012nw,Franco:2013nwa,Bai:2014cna} and the amplituhedron \cite{Arkani-Hamed:2013jha,Arkani-Hamed:2013kca,Franco:2014csa}. Many of the results extend also to 3d ${\N=6}$ ABJM theory \cite{Aharony:2008ug,Hosomichi:2008jb} where several novel properties have emerged \cite{Huang:2013owa,Huang:2014xza,Elvang:2014fja}.

This paper will focus on planar ${\N=4}$ SYM theory, where the ${n}$-particle N${^k}$MHV tree amplitude can be obtained by evaluating the following integral:
\begin{align}
\A_n^{(k)} = \A_n^\MHV \oint_\Gamma \frac{d^{k \times n}C}{GL(k)\,M_1 M_2\ldots M_n}\delta^{4k|4k}\big(C\cdot \Z\big).
\label{integral}
\end{align}
The prefactor ${\A_n^\MHV}$ is the ${n}$-particle MHV amplitude, ${C}$ is a ${k\times n}$ matrix of full rank, and ${\Z}$ encodes the external data (momentum, particle type, etc.) as momentum twistors. In the denominator of the measure, ${M_i}$ is the ${i\th}$ consecutive ${k}$-minor of ${C}$, which means that it is the determinant of the submatrix of ${C}$ with ordered columns ${i,i+1,\ldots, i+k-1 \; (\mod n)}$. The contour on which the integral should be evaluated is designated by ${\Gamma}$. The result is a sum over residues computed at poles of the integrand. There are generally many families of contours which produce equivalent representations of the amplitude due to residue theorems.

Each contour can be thought of as a product of circles wrapping around poles of the integrand, the points where certain minors vanish. The minors are degree-${k}$ polynomials in the variables of integration, so the pole structure is generally very difficult to describe in the formulation \eqref{integral}. However, a technique was presented in \cite{ArkaniHamed:2012nw} for generating charts on the space such that any individual codimension-1 residue occurs at a simple logarithmic singularity. Using one of those charts, the measure on a ${d}$-dimensional submanifold can be decomposed into a product of ${{d\a}/{\a}=\dlog\a}$ quantities. We call this measure ${\omega}$ for future reference:
\beq
\omega=\frac{d^{k \times n}C}{GL(k)\,M_1 M_2\ldots M_n} \to \dlog\a_d \wedge \dlog\a_{d-1}  \wedge \ldots \wedge \dlog\a_1.
\label{firstdlog}
\eeq
Advantages to this formulation are that such charts are easy to generate and that every codimension-1 residue can be reached as a dlog singularity using only a small atlas of charts. A potential disadvantage is that directly relating two distinct charts is non-trivial; this could lead to sign ambiguities when combining individual residues into the amplitude. This is especially important because the residues contain non-physical singularities that should cancel in the tree amplitude sum. Previously, such divergences were shown to appear in pairs, so they at least cancel mod 2 \cite{ArkaniHamed:2012nw}. In Section \ref{sec:Apps} of this paper, we demonstrate that the cancellation is exact,\footnote{The relative signs between NMHV residues were derived in \cite{Elvang:2014fja}, but that method does not easily generalize to higher-${k}$ amplitudes.} which follows from the main result of this paper: the \ref{alg:master} Algorithm introduced in Section \ref{sec:alg}.

\vspace{4mm}\noindent\textbf{\large Summary of Results\vspace{2mm}\\}
The key development of this paper is a systematic algorithm that generates the relative sign between any two BCFW charts on a submanifold, or \emph{cell}, of the Grassmannian. Each chart is defined by a sequence of transpositions ${(a_1b_1)(a_2b_2)\ldots (a_db_d)}$ acting on a permutation labeling a 0-dimensional cell. Equivalently, each chart can be represented by a path through the poset (partially ordered set) of Grassmannian cells. 
Every transposition ${(a_i b_i)}$ corresponds to a factor of ${\dlog \a_i}$ in \eqref{firstdlog}.

To get a sense of the result, it is illustrative to consider the simple example of two charts defined by the sequences ${(ab)(cd)}$ and ${(cd)(ab)}$ with ${a<b<c<d<a\!+\!n}$. In the poset of cells, these sequences define distinct paths from a 0-dimensional cell labeled by ${\s_0}$ to a 2-dimensional cell labeled by ${\s}$, shown in \eqref{egabcd} with edges labeled by the corresponding transpositions:
\als{\includegraphics[height=2.5cm]{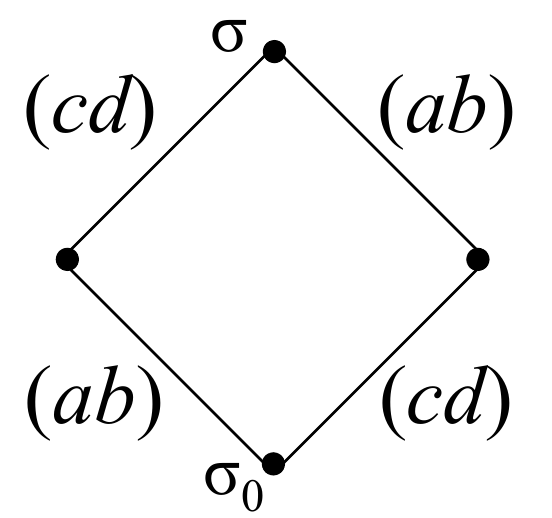}\label{egabcd}}
If we associate the coordinate ${\a_1}$ with ${(ab)}$ and ${\a_2}$ with ${(cd)}$, then the dlog forms generated by the sequences are, respectively,
\als{\omega = \dlog \a_1 \wedge \dlog \a_2, \quad \text{and} \quad \omega'=\dlog \a_2 \wedge \dlog \a_1.}
Clearly the two forms differ by an overall sign, so the two coordinate charts are oppositely oriented. We can encode this property in the poset by weighting the edges with ${\pm 1}$ such that the product of the edge weights around the loop in \eqref{egabcd} is ${-1}$. One choice of suitable signs is
\als{\includegraphics[height=2.5cm]{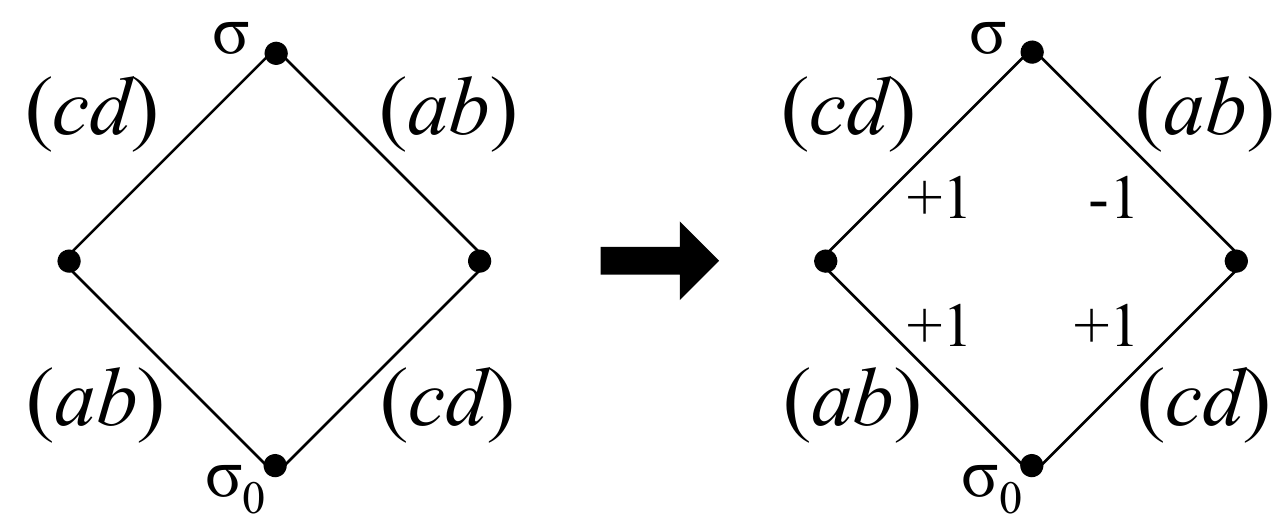}\label{egabcdweighted}}
Then the relative orientation is given by the product of edge signs along the two paths.

We show in Section \ref{sec:master} that the relative orientation between \emph{any} two charts is equal to the product of edge signs around a closed loop in the poset. All the edges can be weighted such that the product of signs around every quadrilateral is ${-1}$, just as in the example \eqref{egabcdweighted}.\footnote{A technique developed by Thomas Lam and David Speyer for assigning the edge weights is presented in Appendix \ref{app:edgesigns} \cite{TLDSprivate}.} The closed loop is obtained by concatenating each input path (blue and green solid lines in (\ref{intropaths}\,A)) together with a sawtooth path connecting their 0d cells (thick red line in (\ref{intropaths}\,A)) times ${-1}$ for each 1d cell along the connecting path. We call this method the \ref{alg:master} Algorithm.

To prove that the \ref{alg:master} Algorithm correctly yields the relative orientation, we introduce two preliminary algorithms in Section \ref{sec:BCFWstand}. When the two paths meet at a common 0d cell, illustrated in (\ref{intropaths}\,B) with blue and green solid lines, Algorithm \ref{alg:alg1} splits the big loop into smaller loops, e.g. using the dashed black lines in (\ref{intropaths}\,B), for which the relative orientations can be computed directly from the corresponding plabic networks. If the paths end in different 0d cells, e.g. (\ref{intropaths}\,C), Algorithm \ref{alg:alg2} additionally computes the relative sign between the source cells.\\
\vspace*{-5mm}
\als{
\begin{array}{ccccc}
\includegraphics[height=5.4cm]{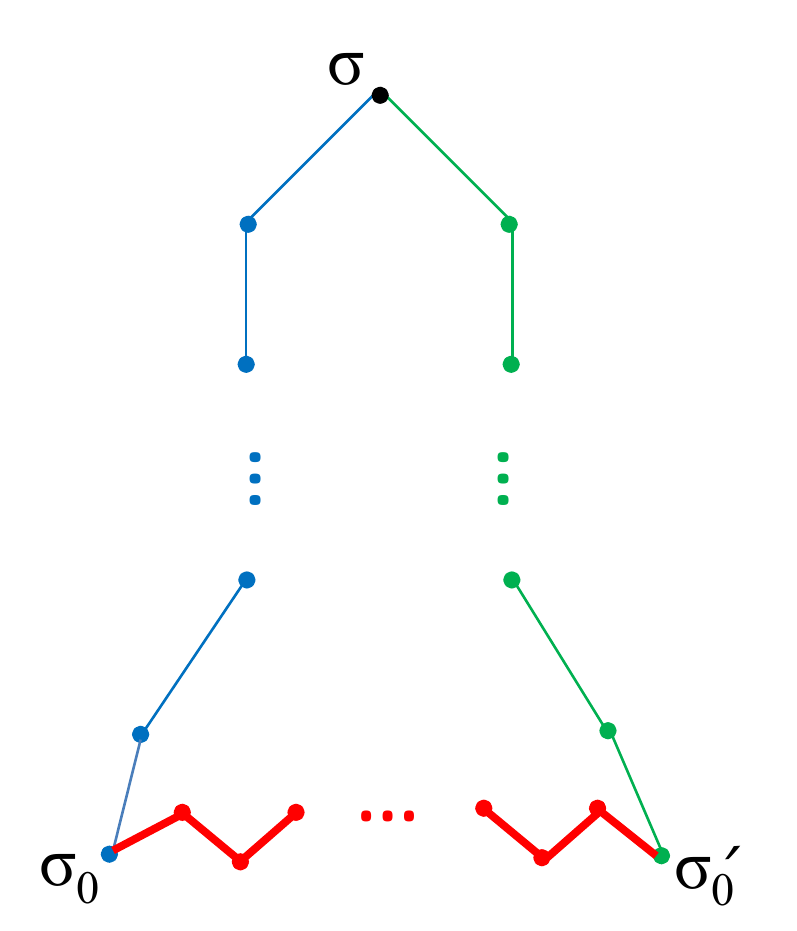} &&
\includegraphics[height=5.4cm]{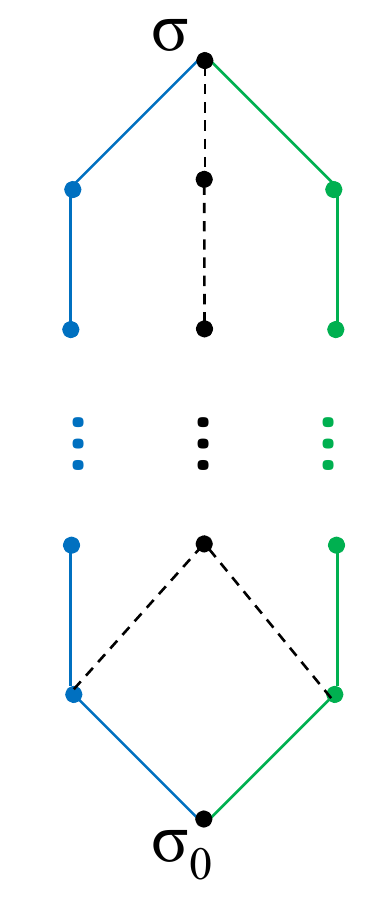} &&
\includegraphics[height=5.4cm]{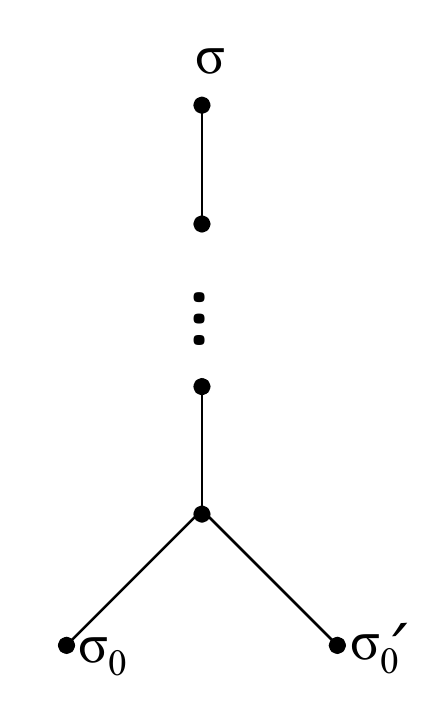} \\
\text{(A) Concatenated paths} && \text{(B) Shared 0d cell} && \text{(C) Different 0d cells}
\end{array}
\label{intropaths}}
We then use Algorithms \ref{alg:alg1} and \ref{alg:alg2} in Section \ref{sec:BCFWgen} to demonstrate that the big loop (\ref{intropaths}\,A) can always be split up into quadrilaterals, each of which contributes a factor of ${-1}$ to the overall sign. Finally in Section \ref{sec:master} it is shown that if the edges are appropriately decorated with ${\pm1}$, then only those that make up the loop  contribute to the end result (in addition to the ${-1}$ from each 1d cell along the sawtooth path).

This eliminates the sign ambiguity between distinct residues, and in Section \ref{sec:Apps} we show that residues contributing to the tree amplitude are always decorated with signs so that all non-physical singularities cancel exactly in the sum of residues. A similar argument demonstrates the existence of a boundary operator that manifestly squares to zero. In addition, we present an interpretation for paths corresponding to decompositions with non-adjacent transpositions. We review the necessary background in Section \ref{sec:background} before presenting the algorithm in Section \ref{sec:alg}, and a few details are relegated to appendices.

During the development of this analysis, a different method for determining relative orientations was proposed independently by Jacob Bourjaily and Alexander Postnikov using determinants of certain large matrices. We computed the orientations of 500 distinct charts on the 10d cell ${\{5, 3, 8, 9, 6, 7, 12, 10, 14, 11\}\in \Gr(3,10)}$ with both methods and found perfect agreement \cite{JBprivate}. The algorithm presented in this paper has also been verified by checking a variety of charts whose orientations are known by other methods. The cancellation of all doubly-appearing poles in the tree amplitude has been confirmed explicitly for all ${n=5,\ldots,13}$ and ${k=3,\ldots,\lfloor n/2 \rfloor}$. All of the algorithms described here have been implemented in Mathematica as an extension of the \emph{positroids} package included with the arXiv version of \cite{ArkaniHamed:2012nw}. Copies of the new code can be provided upon request.

\section{Background}
\label{sec:background}
\subsection{Positroid Stratification}
The positroid stratification is a decomposition of the Grassmannian ${\Gr(k,n)}$ into submanifolds called \emph{positroid cells} (or just \emph{cells}). Cells can be classified according to the ranks of submatrices constructed out of cyclically consecutive columns of a representative matrix.
The \emph{top cell} of ${\Gr(k,n)}$ is the unique cell of highest dimension, ${d=k(n-k)}$. All maximal (${k\times k}$) minors are non-vanishing in the top cell, so all chains of consecutive columns will be full rank. For example, a representative matrix of the top cell of ${\Gr(2,4)}$ is
\als{
C = \left(\begin{array}{cccc}
1 & 0 & c_{13} & c_{14} \\
0 & 1 & c_{23} & c_{24}
\end{array}\right),\label{gr24}
}
where ${c_{ij}\in \C}$ are coordinates on the manifold such that none of the ${2\times2}$ minors vanish. All four parameters must be fixed to specify a point in the top cell, which is consistent with the expected dimension ${d=2(4-2)=4}$.

Lower dimensional cells are reached by fixing relations among the entries so that additional linear dependencies arise among consecutive columns. From a given cell, the accessible codimension-1 submanifolds are called the \emph{boundaries} of that cell. The extra linear relations imply that various minors vanish in the measure of \eqref{integral}. Thus going to the boundary should be interpreted as taking a residue at the corresponding pole. Note that choosing a particular representative matrix amounts to selecting a chart on the cell, so some poles may not be accessible in certain charts. In the example above, the boundary where columns ${\vec{c}_2}$ and ${\vec{c}_3}$ are parallel is accessible by setting ${c_{13}=0}$, but there is no way to reach the boundary where columns ${\vec{c}_1}$ and ${\vec{c}_2}$ are parallel in the stated chart. The latter boundary could be accessed in a ${GL(2)}$-equivalent chart where a different set of columns were set to the identity.

A partial order can be defined on the set of positroid cells by setting ${C \prec C'}$ whenever ${C}$ is a codimension-1 boundary of ${C'}$ \cite{2011arXiv1111.3660K}. The poset structure is interesting for a number of reasons, several of which will be mentioned throughout the text. One such reason is that the poset of positroid cells in ${\Gr(k,n)}$ is isomorphic to a poset of \emph{decorated permutations}. Decorated permutations are similar to standard permutations of the numbers ${1,\ldots,n}$, but differ in that ${k}$ of the entries are shifted forward by ${n}$. To simplify the notation, we will use often `permutation' to mean `decorated permutation' since only the latter are relevant to us.
Permutations will be written single-line notation using curly brackets; an example is given in \eqref{egperm}.
When referencing specific elements of a permutation, we will use the notation ${\s(i)}$ to mean the ${i\th}$ element of the permutation ${\s}$, and with the understanding that ${\s(i\!+\!n)=\s(i)\!+\!n}$. A decorated permutation encodes the ranks of cyclically consecutive submatrices by recording, for each column ${\vec{c}_a}$, the first column ${c_b}$ with ${a\leq b \leq a\!+\!n}$ such that ${\vec{c}_a\in \text{span}(\vec{c}_{a+1},\vec{c}_{a+2},\ldots,\vec{c}_b)}$. The first inequality is saturated when ${\vec{c}_a={0}}$, and the second is saturated when ${\vec{c}_a}$ is linearly independent of all other columns. Continuing with the example \eqref{gr24}, the top cell of ${\Gr(2,4)}$ corresponds to the permutation 
\als{\s_{\text{top}}=\{3,4,5,6\},\label{egperm}}
which says that ${1\to3}$, ${2\to4}$, ${3\to5\equiv 1}$, and ${4\to6 \equiv 2}$. In terms of the linear dependencies among columns of a representative matrix, it means ${\vec{c}_{1} \in \text{span} (\vec{c}_2,\vec{c}_3)}$, ${\vec{c}_2 \in \text{span}(\vec{c}_3,\vec{c}_4)}$, ${\vec{c}_3 \in \text{span}(\vec{c}_4,\vec{c}_5)\equiv \text{span}(\vec{c}_4,\vec{c}_1)}$, and ${\vec{c}_4 \in \text{span}(\vec{c}_5,\vec{c}_6)\equiv \text{span}(\vec{c}_1,\vec{c}_2)}$.

Going to the boundary of a cell involves changing the linear relations among consecutive columns, so in permutation language, the boundary is accessed by exchanging two entries in ${\s}$. The transposition operation that swaps the elements at positions ${a}$ and ${b}$ is denoted by ${(ab)}$, with ${a<b<a\!+\!n}$; we call this the \emph{boundary operation}. 
Note that ${\s}$ only has positions ${1,\ldots,n}$, but ${b}$ can be greater than ${n}$, To account for this, we use
\als{\s'(a)=\s(b) = \s(b\!-\!n)+n \quad \text{and} \quad \s'(b)=\s(a) \Rightarrow \s'(b\!-\!n)=\s(a)\!-\!n.}
In the example \eqref{gr24}, taking the boundary where ${c_{13}=0}$ lands in a cell with linear dependencies specified by
\als{\s'=\{4,3,5,6\}.}
This corresponds to exchanging the first and second elements of ${\s_{\text{top}}}$, i.e.~the transposition ${(1\,2)}$.

Not all exchanges are allowed; for instance, applying the same transposition twice would revert back to the initial cell, which is clearly not a boundary. The allowed boundary transpositions satisfy the following criteria:
\als{
a<b\leq \s(a) < \s(b) \leq a\!+\!n \quad \text{and} \quad \s(q)\not\in(\s(a),\s(b))  ~~ \forall q \in (a,b),
\label{boundary}
}
where ${(a,b)}$ means the set ${\{a+1,b+2,\ldots b-1\}}$. We will call ${(ab)}$ an \emph{adjacent} transposition if all ${q\in(a,b)}$ satisfy ${\s(q)\equiv q ~\mod n}$ \cite{ArkaniHamed:2012nw}. The reason for this distinction will become clear in the next section. A \emph{strictly adjacent} transposition is of the form ${(a\;a\!+\!1)}$. For notational purposes, we define ${(ab)}$ to act on the right, so if ${\s}$ is a boundary of ${\ts}$, then we write ${\s=\ts\cdot(ab)}$. Of course, ${(ab)}$ is its own inverse, so acting on the right with ${(ab)}$ again yields the \emph{inverse boundary operation} ${\s\cdot(ab)=\ts}$. The boundary operation reduces the dimension by one, while the inverse boundary operation increases the dimension by one. Although the notation is identical, it should be clear what we mean from the context. Taking additional (inverse) boundaries leads to expressions like ${\rho=\s\cdot(a_1 b_1)(a_2b_2)\ldots}$, which means first exchange ${\s(a_1)}$ and ${\s(b_1)}$, then swap the elements at positions ${a_2}$ and ${b_2}$, etc.\footnote{To avoid any confusion with accessing elements of the permutation, e.g. ${\s(i)}$, we use the ${\cdot}$ after ${\s}$ to indicate that ${\s}$ is a permutation, and ${(ab)}$ is a transposition. The ${\cdot}$ operation is implicit between neighboring transpositions.}

\subsection{Plabic Graphs}
Permutations and positroid cells can be represented diagrammatically with \emph{plabic} (planar-bicolored) \emph{graphs} and \emph{plabic networks}.\footnote{Physical interpretations of certain plabic networks, also known as `on-shell diagrams,' have been explored in several recent papers, including \cite{ArkaniHamed:2012nw,Bai:2014cna}.} \emph{Plabic graphs} are planar graphs embedded in a disk, in which each vertex is colored either black or white. Any bicolored graph can be made bipartite by adding oppositely-colored bivalent vertices on edges between two identically-colored vertices, so we will assume all graphs have been made bipartite. Some edges are attached to the boundary; we will call these \emph{external legs} and number them in clockwise order ${1,2,\ldots,n}$. If a monovalent leaf is attached to the boundary, it will be called a \emph{lollipop} together with its edge. 
\emph{Plabic networks} are plabic graphs together with weights ${(t_1,t_2,\ldots,t_e)}$ assigned to the edges. The weights are related to coordinates on the corresponding cell, as will be discussed in the next subsection.

Given a plabic graph ${G}$, one can define a \emph{trip permutation} by starting from an external leg ${a}$ and traversing the graph, turning (maximally) left at each white vertex and (maximally) right at each black vertex until the path returns to the boundary at some vertex ${b}$. When ${a\neq b}$, the permutation associated with this trip has 
\als{\s_G(a)= \left\{\begin{array}{ll} b & b>a\\ b\!+\!n & b<a\end{array}\right..}
We will explain the case ${a=b}$ momentarily.

For tree amplitudes, we will be concerned with \emph{reduced} plabic graphs. A plabic graph is reduced if it satisfies the following after deleting all lollipops \cite{TLnotes}:
\begin{enumerate}[nolistsep]
\item It has no leaves;
\item No trip is a cycle;
\item No trip uses a single edge twice;
\item No two trips share two edges ${e_1}$ and ${e_2}$ in the same order.
\end{enumerate}
It follows that any external leg ${a}$ for which ${\s_G(a)\equiv a ~\mod n}$ must be a lollipop \cite{TLnotes}. Specifically, we define ${\s_G(a)=a}$ to be a black lollipop, and ${\s_G(a)=a\!+\!n}$ to be a white lollipop. 

Thus each reduced plabic graph/network corresponds to a unique decorated permutation. However, the correspondence is not a bijection; rather each permutation labels a family of reduced plabic graphs/networks. Members of each family are related by \emph{equivalence moves} that modify the edge weights but leave the permutation unchanged \cite{ArkaniHamed:2012nw,Postnikov:2006kva}:
\newpage
\begin{enumerate}[label=\bfseries(E\arabic*)]
\item \textbf{${\bm{GL(1)}}$ rotation:} At any vertex, one can perform a ${GL(1)}$ rotation that uniformly scales the weights on every attached edge, e.g. for a scaling factor ${f}$,
\als{\includegraphics[height=3cm]{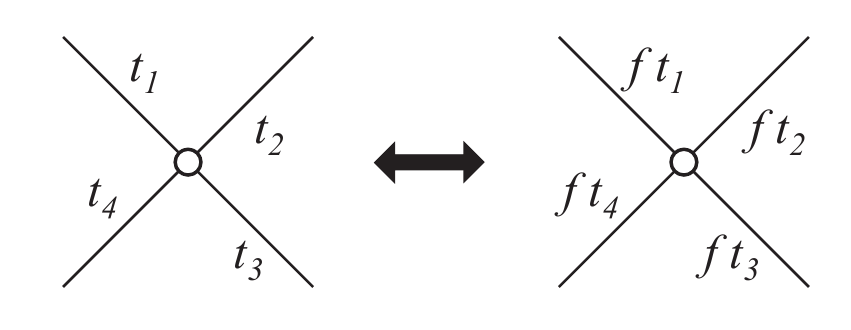}\label{e1}}
\item \textbf{Merge/delete:} Any bivalent vertex whose edges both have weight ${1}$ can be eliminated by merging its neighbors into one combined vertex and deleting the bivalent vertex and its edges, e.g.
\vspace*{-5mm}\als{\includegraphics[height=1.5cm]{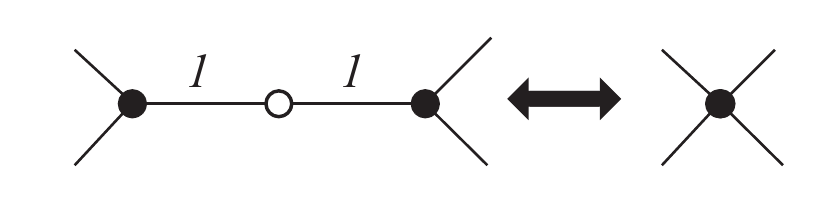}\label{e2}}
If one of its neighbors is the boundary, then the bivalent vertex should be merged with the boundary instead. The inverse operation can also be used to `unmerge' a vertex or boundary.
\item \textbf{Square move:} A four-vertex square with one pattern of coloring is equivalent to the four-vertex square with opposite coloring, e.g. 
\al{\includegraphics[height=3cm]{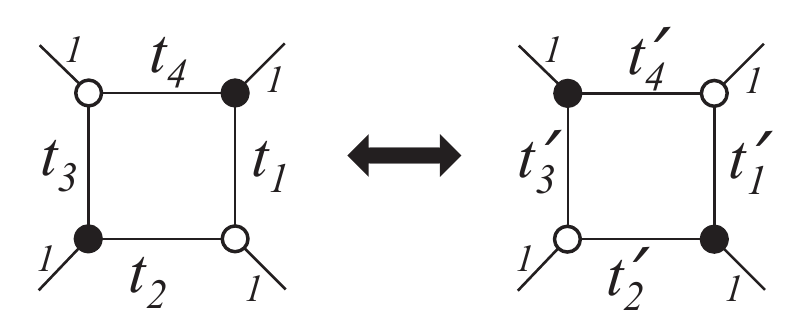}\label{e3}}
The edge weights are related (using the sign conventions of \cite{Franco:2013nwa}):
\als{
\begin{array}{lclclcl}
t_1' = \dfrac{t_3}{t_1 t_3 + t_2 t_4}, && t_2'= \dfrac{t_4}{t_1 t_3 + t_2 t_4},  &&
t_3'= \dfrac{t_1}{t_1 t_3 + t_2 t_4},  &&
t_4= \dfrac{t_2}{t_1 t_3 + t_2 t_4} 
\end{array}
.}
\end{enumerate}

From the plabic network, one can also read off a representative ${k\times n}$ matrix for the corresponding positroid cell, which is called the \emph{boundary measurement matrix} (or \emph{boundary measurement map}). Several related methods exist for defining boundary measurements such as using perfect orientations \cite{ArkaniHamed:2012nw,Franco:2013nwa,Postnikov:2006kva} or perfect matchings \cite{Franco:2013nwa,TLnotes,RK}. The equivalence moves above can be derived by requiring that the boundary measurements are unchanged by each transformation. In the next subsection, we will show how to construct the boundary measurements systematically for the types of plabic networks relevant to amplitudes. We refer the reader to the above references for more details, though we caution that the transformation rules are presented slightly differently.

These three moves are sufficient to transform any plabic network into any other in its equivalence class. In addition, using these moves one can always fix all but ${d}$ of the edge weights to unity in any network that represents a ${d}$-dimensional cell; throughout the rest of this paper, edges with unspecified weights will have weight 1. As we will see shortly, certain planar networks lead to especially convenient paramaterizations of positroid cells. The equivalence moves will allow us to compare the resulting oriented forms.

\subsection{Bridge Decompositions and Charts}
To write down a coordinate chart on a given cell, it is sufficient to construct a representative matrix with the appropriate linear dependencies among its columns. However, in a general chart the boundary structure could be very difficult to identify. Fortunately, there is a straightforward method to construct charts with simple dlog forms as in \eqref{firstdlog} \cite{ArkaniHamed:2012nw}; we review this technique below and how it relates to paths in the poset of cells. Some paths through the poset do not correspond to any such charts, so we suggest an interpretation for those paths in Section \ref{sec:genpaths} and explain the consequences for residues in Section \ref{sec:comparing}.

\subsubsection{Standard BCFW Bridge Decompositions}
\label{sec:standCharts}
Positroid cells of dimension zero in ${\Gr(k,n)}$ correspond to unique plabic graphs made solely out of lollipops with edge weight 1. The ${k}$ legs with ${\s(a)=a\!+\!n}$ have white vertices while the rest are black. The boundary measurement matrix is zero everywhere except the submatrix composed out of the ${k}$ columns corresponding to the white vertices, which together form a ${k\times k}$ identity matrix. There are no degrees of freedom, so the differential form \eqref{firstdlog} is trivial, ${\omega=1}$. 

Since there is a unique representative plabic network for each 0d cell, we will build higher-dimensional representatives out of the set of 0d networks. In the poset, higher-dimensional cells can be reached from 0d cells by repeatedly applying the inverse boundary operation defined below \eqref{boundary}. Equivalently, a ${d}$-dimensional cell can be decomposed into a sequence of adjacent transpositions acting on a 0d permutation, e.g.\footnote{The transposition ${(24)}$ is adjacent because it acts on a permutation whose third element is self-identified.}
\als{\{3,5,4,6\}=\{5,6,3,4\}\cdot(24)(23)(12)\label{eg3546}.}
This is called a \emph{BCFW decomposition} and leads to a convenient graphical representation. In a planar network, an adjacent transposition ${(a_i b_i)}$ amounts to simply adding a white vertex on leg ${a_i}$, a black vertex on leg ${b_i}$, and an edge between them with weight ${\a_i}$:
\als{\includegraphics[height=3cm]{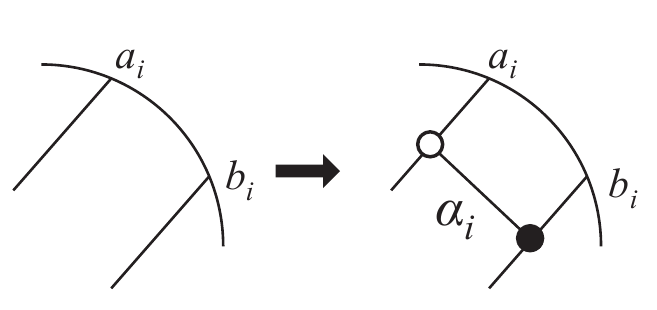}}
This is called a \emph{(BCFW) bridge}. Note that if one of the legs is initially a lollipop, then the resulting leaf should be deleted after adding the bridge \cite{TLnotes,Postnikov:2006kva}.
The example \eqref{eg3546} generates the following sequence of graphs:
\als{\includegraphics[height=3cm]{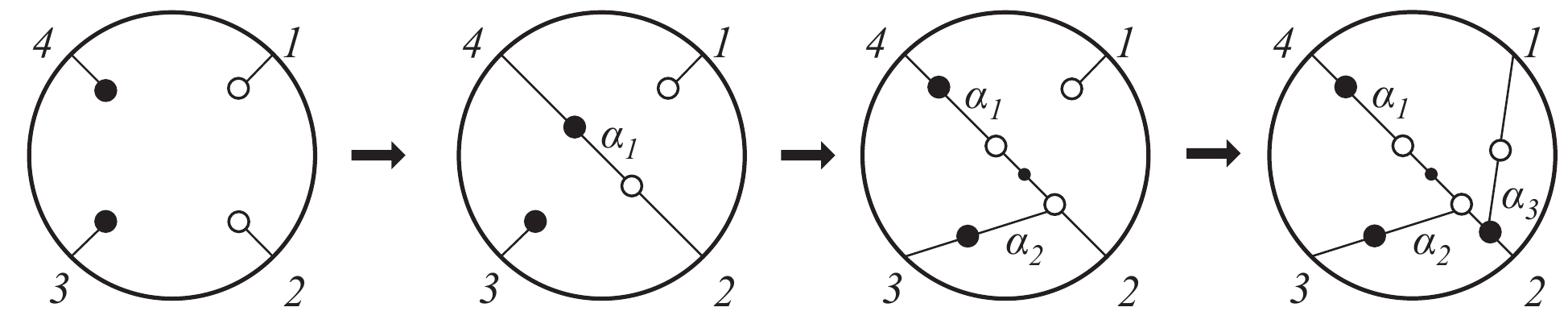}\label{pgrep}}
We have added a black vertex between the first two bridges to make the graph bipartite; it is drawn slightly smaller to distinguish it from the bridge vertices. Many simplifications are possible using the equivalence moves \eqref{e1}-\eqref{e3}.

Adding a BCFW bridge affects the trip permutation by exchanging ${\s_G(a)\leftrightarrow\s_G(b)}$ as desired, and the boundary measurement matrix transforms in a simple way \cite{ArkaniHamed:2012nw},
\al{c_b \to c_b+\alpha_i c_a.}
One can easily check that the linear dependencies of the shifted matrix agree with the expected permutation as long as ${(a_i b_i)}$ is an adjacent transposition. If ${\omega_{i-1}}$ is the differential form associated with the initial cell, then after adding the bridge, the new form is 
\al{\omega_i = \dlog\a_i \wedge \omega_{i-1}.}
This prescription provides a robust way to generate coordinates on any cell in ${\Gr(k,n)}$.

There are generally many ways to decompose a ${d}$-dimensional permutation ${\s}$ into a sequence of adjacent transpositions acting on a 0d cell. In fact, no single chart covers all boundaries of a generic cell. However, an atlas of at most ${n}$ standard BCFW charts is sufficient to cover all boundaries \cite{ArkaniHamed:2012nw}. Every such chart defines a unique path through the poset from ${\s}$ to some ${\s_0}$. In the example \eqref{eg3546}, the path was
\al{ \s=\{3,5,4,6\} \xrightarrow{(12)}  \{5,3,4,6\} \xrightarrow{(23)}\{5,4,3,6\} \xrightarrow{(24)} \{5,6,3,4\}=\s_0. }
Not all BCFW decompositions of ${\s}$ end in the same 0d cell, however. Another BCFW decomposition of ${\{3,5,4,6\}}$ ends in ${\s'_0=\{5,2,7,4\}}$:
\al{\s=\{3,5,4,6\} \xrightarrow{(12)}  \{5,3,4,6\} \xrightarrow{(46)} \{5,2,4,7\} \xrightarrow{(34)} \{5,2,7,4\}=\s'_0. }
The important point is that every BCFW decomposition corresponds to a unique path of length ${d}$ that starts at ${\s}$ and ends in a 0d cell. The converse is not true; some paths of length ${d}$ that start at ${\s}$ and end in a 0d cell do not correspond to any BCFW decomposition.

\subsubsection{Generalized Decompositions}
\label{sec:genpaths}
When evaluating residues, one will often encounter paths through the poset that do not coincide with any single standard chart. These paths contain edges that represent non-adjacent transpositions. Continuing with the earlier example, the following path ends in the same 0d cell as \eqref{eg3546}, but the first transposition ${(13)}$ crosses a non-self-identified leg:
\al{ \s=\{3,5,4,6\} \xrightarrow{(13)}  \{4,5,3,6\} \xrightarrow{(24)}\{4,6,3,5\} \xrightarrow{(14)} \{5,6,3,4\}=\s_0. }
Nevertheless, this path is certainly a possible route when evaluating residues since every codimension-1 boundary is accessible from some adjacent chart \cite{ArkaniHamed:2012nw}. We could, for instance, take the decomposition from \eqref{eg3546}, and take ${\a_2\to 0}$. It is easy to see from the boundary measurement matrix that this yields the desired linear dependencies among columns:
\al{
\left(\begin{array}{cccc}
1 & \a_3 & 0 & 0 \\
0 & 1 & \a_2 & \a_1
\end{array}\right)
\xrightarrow{\a_2\to0}
\left(\begin{array}{cccc}
1 & \a_3 & 0 & 0 \\
0 & 1 & 0 & \a_1
\end{array}\right).
}

The generalization to any path through the poset is clear; each successive step involves computing the residue at a logarithmic singularity in some chart. Thus every path corresponds to some dlog form, not just the paths for which we have explicit plabic graphical representations. We will call these \emph{generalized decompositions} and their coordinates \emph{generalized charts}. In practice, computing residues along a particular path through the poset can always be done using standard adjacent charts, though it may involve changing coordinates at several steps along the way. As we will see in Section \ref{sec:comparing}, the sign of the resulting residue depends only on the path taken, and not on the choice of reference charts along the way.

\section{Relating Distinct Charts}
\label{sec:alg}
Recall that a decomposition for a ${d}$-dimensional cell ${C}$ corresponds to a sequence of transpositions applied to a permutation ${\s_0}$ labeling a 0-dimensional cell ${C_0}$. Each sequence defines a particular path through the poset with endpoints at ${C}$ and ${C_0}$. In this section, we will show that the relative sign between the dlog forms generated by two distinct charts can be obtained systematically. We will first show this result for standard BCFW charts constructed using only adjacent transpositions. Subsequently, we will discuss the generalized situation where we do not always have a simple graphical representation; the convention for deriving signs will be extended to cover the additional possibilities while maintaining consistency with the standard setup. The extended conventions will also lead to a simpler method for comparing charts.

\subsection{Standard BCFW Charts}
\label{sec:BCFWstand}
There are two cases that we need to address depending on whether the decompositions end in identical 0d cells or distinct ones. We will cover the identical case first and then deal with the other situation.

\subsubsection{Charts with Identical 0d Cells}
We assume first that the 0d cell labeled by ${\sigma_0}$ is the same for both paths. The ${d}$-dimensional cell labeled by ${\s}$ is connected to ${\s_0}$ by two sequences of transpositions
\als{\s=\s_0\cdot(a_1 b_1)(a_2 b_2)\ldots (a_d b_d)=\s_0\cdot(a_1' b_1')(a_2' b_2')\ldots(a_d' b_d').}
Graphically, the concatenation of the two paths creates a closed loop of length ${2d}$ in the poset, shown here schematically with one path denoted by a blue solid line and the other by a green dashed line:
\als{\includegraphics[height=7cm]{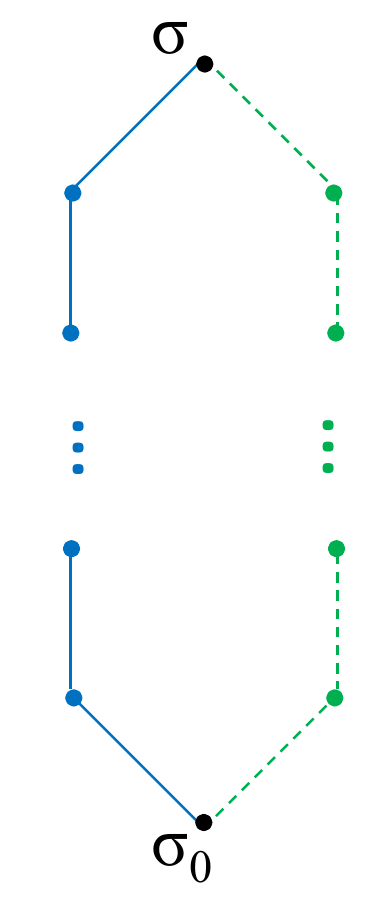}}

The relative orientation of the two dlog forms corresponding to the sequences can be obtained by a simple algorithm. Before presenting that result, we will need the following lemma:
\begin{lemma}
The top cell can be reached from any cell, ${C}$, of dimension ${d}$ by a sequence of ${k(n-k)-d}$ strictly adjacent transpositions.
\label{lemma}
\end{lemma}
\begin{proof}
Let ${\s}$ be the permutation labeling ${C\in\Gr(k,n)}$. When ${\sigma}$ contains two neighboring elements satisfying ${\sigma(i)>\sigma(i+1)}$ (with ${\sigma(n+1)= \sigma(1)\!+\!n}$), this is called an \emph{inversion}. Such an inversion can be removed by applying the (strictly adjacent) transposition ${(i~i+1)}$. Since all entries in ${\sigma_{\text{top}}}$ are ordered, one can reach the top cell by iteratively eliminating all inversions.
\end{proof}
For example, the top cell of ${\Gr(2,6)}$ can be reached from the 5-dimensional cell ${\{2, 3, 4, 6, 7, 11\}}$ by the sequence of transpositions ${(6\,7)(1\,2)(2\,3)}$:
\als{
\{2, 3, 4, 6, 7, 11\}\xrightarrow{(6\,7)}\{5, 3, 4, 6, 7, 8\} \xrightarrow{(1\,2)}\{3, 5, 4, 6, 7, 8\} \xrightarrow{(2\,3)} \{3, 4, 5, 6, 7, 8\} .
}
Using this procedure, any BCFW sequence can be extended to reach the top cell using only strictly adjacent transpositions. Since ${(i,i+1)}$ does not cross any legs, the resulting sequence will be a valid BCFW sequence. Therefore, we may assume without loss of generality that the cell on which we seek to compare orientations is the top cell because two sequences which lead to a cell of lower dimension can be trivially extended to top cell sequences by appending the same transpositions to both paths. This will not affect the relative sign of the forms since both will have identical pieces appended to them.

We turn now to the sign-comparison algorithm. The idea is to compare each BCFW chart to specially chosen reference charts whose relative orientation is easy to compute. They are chosen so that at each iteration, the loop in the poset (initially of length ${2d}$) is shortened. Then the final relative orientation is the product over all the intermediate orientations.

\begin{alg}
\textbf{Input:} Two BCFW sequences of length ${d=k(n-k)}$: ${\w=(a_1 b_1)(a_2 b_2)\ldots (a_d b_d)}$ and
\\\hspace*{2cm} ${\w'=(a_1' b_1')(a_2' b_2')\ldots(a_d' b_d')}$\\
\textbf{Output:} ${\pm 1}$\vspace{2mm}\\
\textbf{Procedure:}\vspace{-5mm}
\begin{enumerate}[label=\bfseries\arabic*)]
\item Let ${j}$ be the smallest index such that ${(a_j b_j)\neq(a_j' b_j')}$. The transpositions with ${i>j}$ yield a closed loop of length ${\ell\leq 2d}$. If there is no such position, then the paths are identical, so return ${+1}$.
\item Let ${\sigma}$ label the ${j}$-dimensional cell reached by the sequence of transpositions ${(a_1 b_1)(a_2 b_2)\ldots (a_j b_j)}$ and ${\sigma'}$ label that reached by ${(a_1' b_1')(a_2' b_2)'\ldots (a_j' b_j')}$. Using the following rules, construct reference charts to which the initial charts should be compared. Comparing the two reference charts produces a known sign; the relevant parts are displayed with each step, and their relative signs are derived in Appendix \ref{app:alg1}. There are several cases to consider (with ${a<b<c<d<a\!+\!n}$):
\begin{enumerate}[label=\bfseries\roman*)]
\item ${\bm{(a_j b_j)=(a b), ~ (a_j' b_j') = (c d)}}$\\
The ${j}$-dimensional cells ${\sigma}$ and ${\sigma'}$ have a shared ${(j+1)}$-dimensional neighbor ${\tilde{\sigma}=\sigma\cdot(cd)=\sigma'\cdot(ab)}$. Let ${\u}$ be the sequence generated by Lemma \ref{lemma} for the cell labeled by ${\tilde{\sigma}}$. Then the reference sequences and relative sign are:\newline
\begin{minipage}{\textwidth}
\begin{minipage}[t]{0.7\textwidth}
\vspace{2mm}
\begin{enumerate}[label=$\bullet$,itemsep=-1mm,itemindent=-6mm]
\item Ref. sequence 1: ${\tw=(a_1 b_1)(a_2 b_2)\ldots (a_{j-1} b_{j-1})(ab)(cd) \u}$
\item Ref. sequence 2: ${\tw'=(a_1' b_1')(a_2' b_2')\ldots (a_{j-1}' b_{j-1}')(cd)(ab) \u}$
\item The relative sign between reference charts is ${-1}$.
\end{enumerate}
\end{minipage}
\begin{minipage}[t]{0.2\linewidth}
\vspace{-5mm}
\begin{figure}[H]\includegraphics[width=3cm]{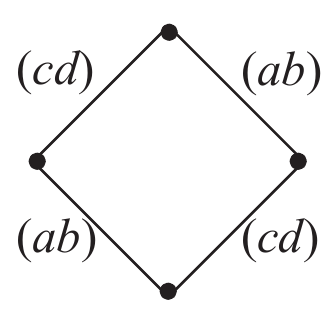}
\end{figure}
\end{minipage}
\vspace{-3mm}
\end{minipage}

\item ${\bm{(a_j b_j)=(a b), ~ (a_j' b_j') = (a c)}}$\\
In this case ${\sigma}$ and ${\sigma'}$ have a shared ${(j+1)}$-dimensional neighbor ${\tilde{\sigma}=\sigma\cdot(bc)=\sigma'\cdot(ab)}$. Let ${\u}$ be the sequence generated by Lemma \ref{lemma} for the cell labeled by ${\tilde{\sigma}}$. Then the reference sequences and relative sign are:\newline
\begin{minipage}{\textwidth}
\begin{minipage}[t]{0.7\textwidth}
\vspace{2mm}
\begin{enumerate}[label=$\bullet$,itemsep=-1mm,itemindent=-6mm]
\item Ref. sequence 1: ${\tw=(a_1 b_1)(a_2 b_2)\ldots (a_{j-1} b_{j-1})(ab)(bc) \u}$
\item Ref. sequence 2: ${\tw'=(a_1' b_1')(a_2' b_2')\ldots (a_{j-1}' b_{j-1}')(ac)(ab) \u}$
\item The relative sign between reference charts is ${-1}$.
\end{enumerate}
\end{minipage}
\begin{minipage}[t]{0.2\linewidth}
\vspace{-5mm}
\begin{figure}[H]\includegraphics[width=3cm]{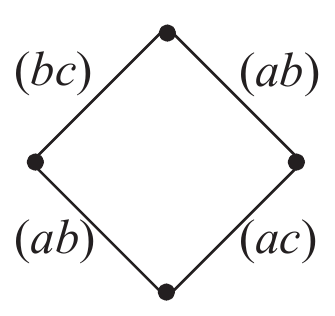}
\end{figure}
\end{minipage}
\vspace{-3mm}
\end{minipage}

\item ${\bm{(a_j b_j)=(a c), ~ (a_j' b_j') = (b c)}}$\\
Again ${\sigma}$ and ${\sigma'}$ have a shared ${(j+1)}$-dimensional neighbor ${\tilde{\sigma}=\sigma\cdot(bc)=\sigma'\cdot(ab)}$. Let ${\u}$ be the sequence generated by Lemma \ref{lemma} for the cell labeled by ${\tilde{\sigma}}$. Then the reference sequences and relative sign are:\newline
\begin{minipage}{\textwidth}
\begin{minipage}[t]{0.7\textwidth}
\vspace{2mm}
\begin{enumerate}[label=$\bullet$,itemsep=-1mm,itemindent=-6mm]
\item Ref. sequence 1: ${\tw=(a_1 b_1)(a_2 b_2)\ldots (a_{j-1} b_{j-1})(ac)(bc) \u}$
\item Ref. sequence 2: ${\tw'=(a_1' b_1')(a_2' b_2')\ldots (a_{j-1}' b_{j-1}')(bc)(ab) \u}$
\item The relative sign between reference charts is ${-1}$.
\end{enumerate}
\end{minipage}
\begin{minipage}[t]{0.2\linewidth}
\vspace{-5mm}
\begin{figure}[H]\includegraphics[width=3cm]{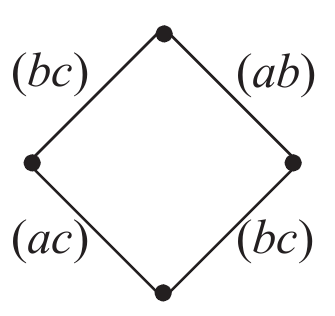}
\end{figure}
\end{minipage}
\vspace{-3mm}
\end{minipage}

\item ${\bm{(a_j b_j)=(a b), ~ (a_j' b_j') = (b c)}}$\\
Using only adjacent transpositions, ${\s}$ and ${\s'}$ do not have a common ${(j+1)}$-dimensional neighbor. However, certain neighbors of ${\s}$ and ${\s'}$ do have a common neighbor of dimension ${(j+2)}$. Specifically, let ${\ts=\s\cdot(bc)}$ and ${\ts'=\s'\cdot(ab)}$. Then ${\ts}$ and ${\ts'}$ have a common ${(j+2)}$-dimensional neighbor ${\rho=\ts\cdot(ab)=\ts'\cdot(bc)}$. Let ${\u}$ be the sequence generated by Lemma \ref{lemma} for the cell labeled by ${\rho}$. Then the reference sequences and relative sign are:\newline
\begin{minipage}{\textwidth}
\begin{minipage}[t]{0.7\textwidth}
\vspace{2mm}
\begin{enumerate}[label=$\bullet$,itemsep=-1mm,itemindent=-6mm]
\item Ref. sequence 1: ${\tw=(a_1 b_1)(a_2 b_2)\ldots (a_{j-1} b_{j-1})(ab)(bc)(ab) \u}$
\item Ref. sequence 2: ${\tw'=(a_1' b_1')(a_2' b_2')\ldots (a_{j-1}' b_{j-1}')(bc)(ab)(bc) \u}$
\item The relative sign between reference charts is ${+1}$.
\end{enumerate}
\end{minipage}
\begin{minipage}[t]{0.2\linewidth}
\vspace{-8mm}
\begin{figure}[H]\hspace*{3mm}\includegraphics[width=2.5cm]{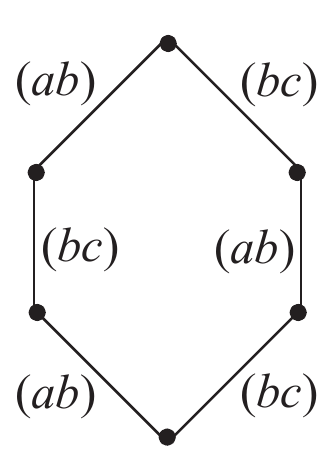}
\end{figure}
\end{minipage}
\vspace{-3mm}
\end{minipage}
\newpage

\item ${\bm{(a_j b_j)=(a c), ~ (a_j' b_j') = (b d)}}$\\
Similar to the previous case, ${\s}$ and ${\s'}$ do not have a common ${(j+1)}$-dimensional neighbor using only adjacent transpositions. Nonetheless, with ${\ts=\s\cdot(bc)}$ and ${\ts'=\s'\cdot(bc)}$, then ${\ts}$ and ${\ts'}$ have a common ${(j+2)}$-dimensional neighbor ${\rho=\ts\cdot(cd)=\ts'\cdot(ab)}$. Let ${\u}$ be the sequence generated by Lemma \ref{lemma} for the cell labeled by ${\rho}$. Then the reference sequences and relative sign are:\newline
\begin{minipage}{\textwidth}
\begin{minipage}[t]{0.7\textwidth}
\vspace{2mm}
\begin{enumerate}[label=$\bullet$,itemsep=-1mm,itemindent=-6mm]
\item Ref. sequence 1: ${\tw=(a_1 b_1)(a_2 b_2)\ldots (a_{j-1} b_{j-1})(ac)(bc)(cd) \u}$
\item Ref. sequence 2: ${\tw'=(a_1' b_1')(a_2' b_2')\ldots (a_{j-1}' b_{j-1}')(bd)(bc)(ab) \u}$
\item The relative sign between reference charts is ${-1}$.
\end{enumerate}
\end{minipage}
\begin{minipage}[t]{0.2\linewidth}
\vspace{-8mm}
\begin{figure}[H]\hspace*{3mm}\includegraphics[width=2.5cm]{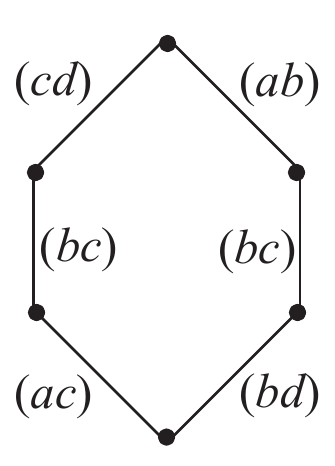}
\end{figure}
\end{minipage}
\vspace{-3mm}
\end{minipage}

\item ${\bm{(a_j b_j)=(b c), ~ (a_j' b_j') = (a d)}}$\\
In this case, we must look further to find a shared cell above ${\sigma}$ and ${\sigma'}$ using only adjacent transpositions. They have a shared ${(j+3)}$-dimensional great-grandparent ${\tilde{\rho}=\sigma\cdot(ab)(cd)(bc)=\sigma'\cdot(bc)(ab)(cd)}$. Let ${\u}$ be the sequence generated by Lemma \ref{lemma} for the cell labeled by ${\tilde{\rho}}$. Then the reference sequences and relative sign are:\newline
\begin{minipage}{\textwidth}
\begin{minipage}[t]{0.7\textwidth}
\vspace{2mm}
\begin{enumerate}[label=$\bullet$,itemsep=-1mm,itemindent=-6mm]
\item Ref. sequence 1: ${\tw=(a_1 b_1)(a_2 b_2)\ldots (a_{j-1} b_{j-1}) (bc)(ab)(cd)(bc)\u}$
\item Ref. sequence 2: ${\tw'=(a_1' b_1')(a_2' b_2')\ldots (a_{j-1}' b_{j-1}')(ad)(bc)(ab)(cd) \u}$
\item The relative sign between reference charts is ${-1}$.
\end{enumerate}
\end{minipage}
\begin{minipage}[t]{0.2\linewidth}
\vspace{-8mm}
\begin{figure}[H]\hspace*{3mm}\includegraphics[width=2.5cm]{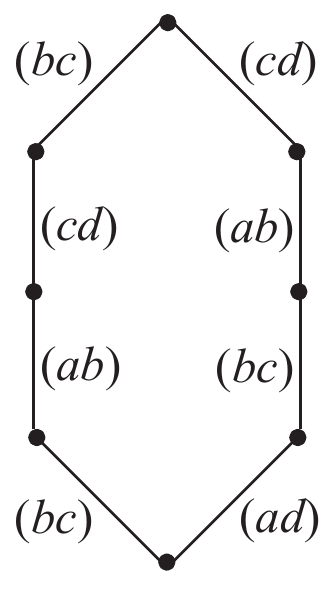}
\end{figure}
\end{minipage}
\end{minipage}

\end{enumerate}
\vspace{-2cm}\item Repeat this algorithm to compare ${\w}$ to ${\tw}$ and ${\w'}$ to ${\tw'}$.
\item Return the product of the relative sign from step (2) times\\ the result of each comparison in step (3).
\end{enumerate}
\label{alg:alg1}
\end{alg}

The relative orientation of any two BCFW charts with identical endpoints can be compared with this algorithm. 
In the second part of the proof below, we show that the signs and reference charts presented in step (2i) are correct; this should also serve as an illustrative example of the algorithm in action.

\begin{proof}
We need to show that the algorithm will terminate in a finite number of iterations, and that the sign generated at each step is correct.
\begin{itemize}
\item We will first show that the algorithm will terminate after a finite number of iterations. The reference sequences constructed in step (2) are chosen so that when the algorithm is called again in step (3) to compare ${\w}$ to ${\tw}$, the new inputs satisfy ${(a_i b_i)=(a_i' b_i')}$ for all ${i\leq j}$. The same is true for the comparison of ${\w'}$ and ${\tw'}$. Step (1) searches for the first point at which the input sequences differ, so by construction, the next position will be at least ${j+1}$, which is larger than in the previous iteration. Since $j$ is bounded by ${d}$, the algorithm will eventually terminate.
\item Next we will explain the results presented in step (2i). Since ${a<b<c<d<a\!+\!n}$, the two transpositions can be applied in either order without violating the adjacent requirement of BCFW sequences. Therefore ${\sigma}$ and ${\sigma'}$ have a common neighbor ${\tilde{\sigma}}$, and both ${(a_1 b_1)(a_2 b_2)\ldots (a_{j-1} b_{j-1})(ab)(cd)\u}$ and ${(a_1' b_1')(a_2' b_2')\ldots (a_{j-1}' b_{j-1}')(cd)(ab)\u}$ are valid BCFW sequences for the top cell. Moreover, since ${(a_i' b_i')=(a_i b_i)}$ for all ${i<j}$, their corresponding forms differ only in positions ${j}$ and ${j+1}$:
\als{
\hspace*{-3mm}\omega &= \dlog\a_d \wedge \dlog \a_{d-1}  \ldots \wedge \dlog\a_{j+2} \wedge \dlog\a_{j+1} \wedge \dlog\a_{j} \wedge \dlog\a_{j-1}\ldots \wedge \dlog\a_1, \\
\hspace*{-3mm}\omega' &= \dlog\a_d \wedge \dlog \a_{d-1}  \ldots \wedge \dlog\a_{j+2} \wedge \dlog\beta_{j+1} \wedge \dlog\beta_{j} \wedge \dlog\a_{j-1}\ldots \wedge \dlog\a_1,
}
where ${\alpha_j}$ and ${\alpha_{j+1}}$ are the weights associated with ${(ab)}$ and ${(cd)}$ in the first sequence, while ${\beta_{j}}$ and ${\beta_{j+1}}$ are associated with ${(cd)}$ and ${(ab)}$ in the second sequence. To determine the relationship between the two forms, we will use the plabic graph representations of the transpositions as BCFW bridges. Focusing on the ${j}$ and ${j+1}$ parts of the graph, we find the following:
\als{
\includegraphics[height=6cm]{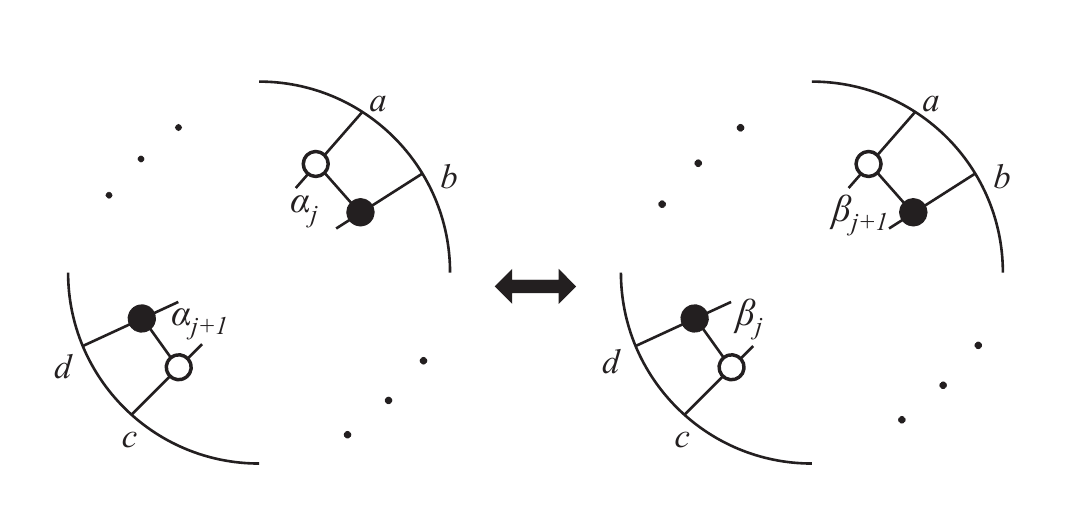}\label{abcd}
}

The left and right diagrams can only be equivalent if ${\b_j=\a_{j+1}}$ and ${\b_{j+1} = \a_{j}}$, which implies
\als{
\omega'&=\dlog\a_d \wedge  \ldots \wedge \dlog\a_{j+2} \wedge \dlog\a_{j} \wedge \dlog\a_{j+1} \wedge \dlog\a_{j-1}\ldots \wedge \dlog\a_1=-\omega.
}
Thus the relative sign between the reference charts is ${-1}$.
\end{itemize}
The remaining cases are described in Appendix \ref{app:alg1}.
\end{proof}
Thus we now have an algorithm that correctly computes the relative signs for BCFW forms generated from identical 0d cells.

\subsubsection{Charts with Distinct 0d Cells}
The next step will be to extend the result to decompositions that terminate in distinct 0d cells. Let us start with the simplest case: finding the relative sign between two charts on a 1d cell labeled by ${\sigma_1}$. Since it is 1-dimensional, all but two entries in ${\sigma}$ are self-identified ${\mod n}$. The two non-trivial positions can be labeled ${a}$ and ${b}$ such that ${a<b<a\!+\!n}$. 

The 1d cell has exactly two boundaries, which can be accessed respectively by ${(ab)}$ or ${(b\;a\!+\!n)}$, so the two BCFW sequences whose charts we will compare are ${(ab)}$ and ${(b\;a\!+\!n)}$:
\als{\includegraphics[height=2cm]{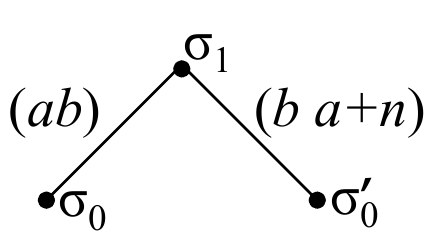}}
It is straightforward to find the relative orientations of the corresponding forms,
\als{\omega_1=\dlog\a \quad \text{and} \quad \omega_1'=\dlog\b,} using basic plabic graph manipulations. As we see in the following diagrams, several ${GL(1)}$ rotations \eqref{e1} combined with a merge and unmerge with the boundary \eqref{e2} are sufficient to discover that ${\b=1/\a}$,
\als{\includegraphics[height=2.5cm]{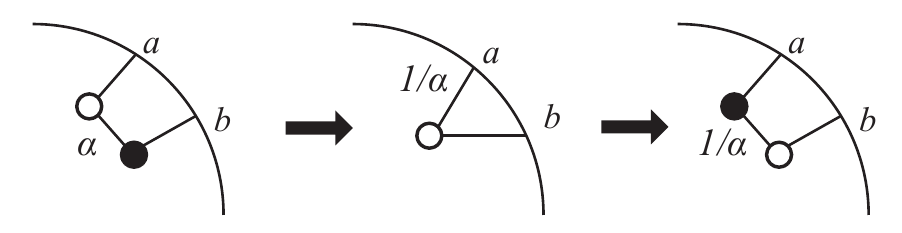}\label{1dbridge}}
Therefore, the forms are oppositely oriented, i.e.~${\omega_1'=-\omega_1}$.

More generally, we could consider charts on ${d}$-dimensional cells whose corresponding BCFW sequences are identical everywhere except for the first transposition:
\als{\includegraphics[height=5cm]{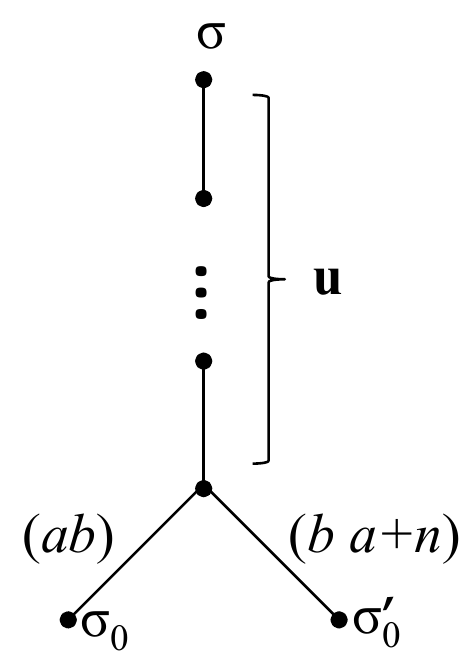}\label{split}}
The associated forms are
\als{
\omega_d &= \dlog\a_d \wedge \dlog\a_{d-1} \wedge \ldots \wedge \dlog\a_1,\\
\omega_d' &= \dlog\a_d \wedge \dlog\a_{d-1} \wedge \ldots \wedge \dlog\b_1.}
For example, we could find a Lemma \ref{lemma} sequence, ${\u}$, for the 1d cell above and compare the two charts defined by ${\s_0\cdot(ab)\u}$ and ${\s_0'\cdot(b\,a\!+\!n)\u}$. Since BCFW sequences such as ${\u}$ use only adjacent transpositions, no bridge will ever be attached to legs ${a}$ and ${b}$ further from the boundary than the first bridge ${(ab)}$, resp., ${(b\;a\!+\!n)}$. Therefore, one can apply very similar logic as in the 1d case\footnote{The only difference being that the merge/unmerge moves may be with other vertices instead of the boundary.} to find that ${\beta_1=1/\alpha_1}$, so ${\omega_d'=-\omega_d}$.

We can easily extend this to any two charts whose 0d endpoints share a common 1d neighbor, ${\sigma_1}$, by applying Algorithm \ref{alg:alg1}. Each sequence can be related to a reference sequence with the same 0d cell, but which goes through ${\sigma_1}$ and then follows some arbitrarily chosen path, say ${\u}$, to the top cell. If the same path is chosen to compare to both charts, then the relative orientation of the reference charts is ${-1}$.

Finally, this can be extended to any two charts terminating in arbitrarily separated 0d cells by iterating the previous step together with Algorithm \ref{alg:alg1}. We combine this into Algorithm \ref{alg:alg2}.
\begin{alg}
\label{alg:alg2}
\textbf{Input:} Two BCFW sequences of length ${d=k(n-k)}$: ${\w=(a_1 b_1)(a_2 b_2)\ldots (a_d b_d)}$ and
\\\hspace*{2cm} ${\w'=(a_1' b_1')(a_2' b_2')\ldots(a_d' b_d')}$; and their 0d endpoints: ${\sigma_0}$ and ${\sigma_0'}$.\\
\textbf{Output:} ${\pm 1}$\vspace{2mm}\\
\textbf{Procedure:}\vspace{-5mm}
\begin{enumerate}[label=\bfseries\arabic*)]
\item If ${\sigma_0=\sigma_0'}$, compute the relative sign of the two forms using Algorithm 1. Return the result.
\item Else, let ${a}$ be the smallest index such that ${\sigma_0(a)<\sigma_0'(a)}$, and let ${b}$ be the smallest index such that ${\sigma_0(b)>\sigma_0'(b)}$. We assume that ${a<b}$; if not, then the roles of $\s_0$ and $\s'_0$ should be exchanged. Let ${\ts=\sigma_0\cdot(ab)}$, whose boundaries are ${\s_0}$ and ${\ts_0=\ts\cdot(b\;a\!+\!n)}$, and define ${\u}$ to be the Lemma \ref{lemma} sequence from ${\ts}$ to the top cell. Construct two reference sequences: ${\tw = (ab)\u}$ and ${\tw'= (b\;a\!+\!n)\u}$.
\item Repeat this algorithm to compare ${(\w,\s_0)}$ to ${(\tw,\s_0)}$, and ${(\w',\s_0')}$ to ${(\tw',\ts_0)}$. 
\item Return the product of the results from step (3) times ${-1}$ due to the relative sign between ${\tw}$ and  ${\tw'}$.
\end{enumerate}
\end{alg}
\newpage
\begin{proof}
Assuming that the cells and edges in step (2) exist, the sign at each iteration is valid because it uses Algorithm 1 to compare charts with identical 0d cells, and it returns ${-1}$ for each pair of sequences that differ only in the first position. It remains to show that step (2) is correct. Since all entries in the permutations labeling 0d cells are self-identified ${\mod n}$, the definitions of ${a}$ and ${b}$ imply that:
\als{\s_0(a)=a,\quad \s_0(b)=b\!+\!n,\quad \s_0'(a)=a\!+\!n, \text{~~and~~} \s_0'(b)=b.}
There exists another 0d cell ${\ts_0}$, which is identical to ${\s_0}$ except 
\als{\ts_0(a)=\s_0(a)\!+\!n=a\!+\!n=\s_0'(a) \text{~~~and~~~} \ts_0(b)=\s_0(b)\!-\!n=b=\s_0'(b).}
The new cell ${\ts_0}$ has two important properties:
\begin{itemize}
\item The first is that ${\s_0}$ and ${\ts_0}$ have a common 1d neighbor, ${\ts=\s_0\cdot (ab) = \ts_0\cdot(b\;a\!+\!n)}$. Since all other entries in ${\s_0}$ and ${\ts_0}$ are self-identified ${\mod n}$, both ${(ab)}$ and ${(b\;a\!+\!n)}$ are adjacent. Thus the cells and reference sequences of step (2) are uniquely defined and satisfy the standard adjacency requirements.
\item The second is that ${\ts_0}$ differs from ${\s_0'}$ at fewer sites than ${\s_0}$ differs. If there are ${m\geq1}$ locations where ${\s_0(i)-\s_0'(i)\neq0}$, then there are only ${m-2}$ locations where ${\ts_0(i)-\s_0'(i)\neq0}$. Since all entries in 0d cells are self-identified ${\mod n}$, and ${k}$ entries are greater than ${n}$, then ${m}$ must be even and no larger than ${2k}$. Hence, Algorithm \ref{alg:alg2} will complete after at most ${k}$ iterations.
\end{itemize}
\vspace{-1.1cm}\end{proof}
\noindent Therefore Algorithms \ref{alg:alg1} and \ref{alg:alg2} are together sufficient to find the relative orientation of any two standard BCFW charts.
 
\subsection{Generalized Decompositions}
\label{sec:BCFWgen}
The plabic graph representation explained in Section \ref{sec:standCharts} is convenient for the study of standard BCFW charts. They are a subset of the generalized decompositions, so the rules for defining reference charts and relative signs in Algorithms \ref{alg:alg1} and \ref{alg:alg2} will still apply. Since cases (i)-(vi) in step (2) of Algorithm \ref{alg:alg1} cover all possible comparisons, the techniques introduced in the previous section are in fact sufficient to find the relative orientation of \textit{any} two charts.

However, there is an advantage to studying the generalized charts more closely. Cases (iv)-(vi) in Algorithm \ref{alg:alg1} were distinctly different than the other three cases because the reference charts required two or three steps to meet at a common cell as opposed to only one step in the earlier cases. Comparing the reference paths of the first three cases shows that they define quadrilaterals, while cases (iv) and (v) define hexagons, and (vi) defines an octagon. The relevant sections of the poset are depicted in \eqref{shapes} with the solid lines indicating the paths used in Algorithm \ref{alg:alg1}, and the dashed lines showing the additional edges that were not used (the finely dotted lines indicate edges that do not always exist). We include the quadrilaterals from cases (i)-(iii) for completeness:
\als{\includegraphics[height=8cm]{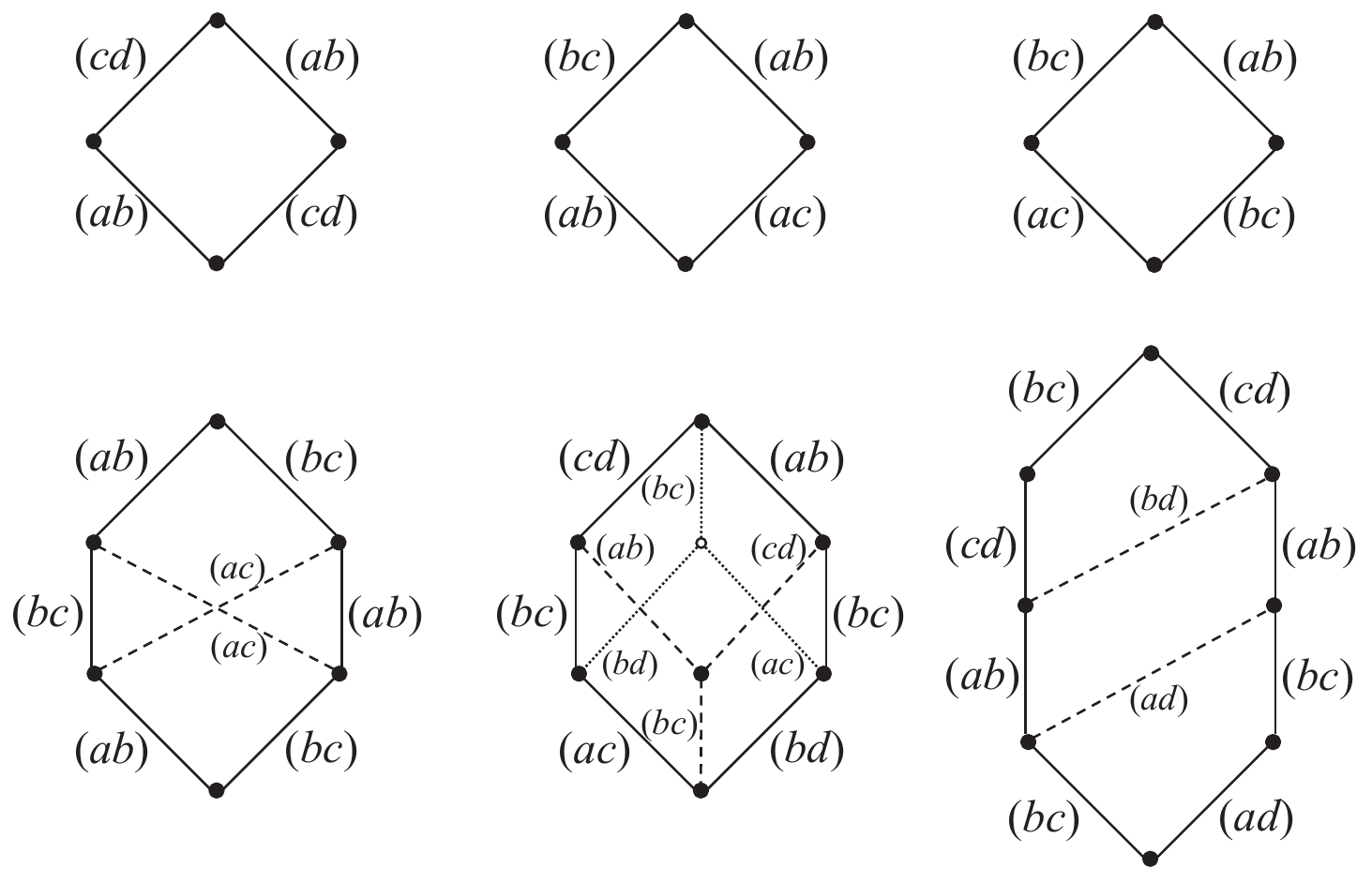}\label{shapes}}
The extra transpositions permitted in generalized charts allow the hexagons and octagons to be refined into quadrilaterals. Some of the internal quadrilaterals are equivalent to those from cases (i)-(iii), but there are also new ones. The relative orientation of two charts which differ only by one of the new quadrilaterals needs to be determined. To fix the signs, we require that the refined polygons produce the same signs as above when split into charts that differ by the interior quadrilaterals. 

The hexagon from case (iv) can be split into a pair of quadrilaterals two ways. Either way, the top quadrilateral appears in one of the first three cases, which implies that the relative orientation around the lower quadrilateral must be ${-1}$ in order to agree with the overall sign of ${+1}$ derived in Appendix \ref{app:alg1}.

In case (v), the edges shown with dashed lines always exist, so the hexagon can be split into three quadrilaterals. The one on the lower left is equivalent to case (iii), the lower right is equivalent to case (ii), and the top quadrilateral is identical to case (i). Hence the product of the three individual signs is ${(-1)^3=-1}$, in agreement with the result found in Appendix \ref{app:alg1}. In some situations, the hexagon can also be split up using the finely dotted lines. Then the top two hexagons are equivalent to cases (ii) and (iii), which implies that the relative orientation around the lower quadrilateral must also be ${-1}$.

Finally, the octagon from case (vi) can be also be split into three quadrilaterals. The top one matches case (iii), and the middle is equivalent to case (ii), so the relative orientation around the lower one must be ${-1}$ to agree with Appendix \ref{alg:alg1}. This is unsurprising considering that the two transpositions are completely disjoint, so applying them in opposite order would suggest that the forms differ by a minus signs, similar to example \eqref{egabcdweighted} in the \ref{sec:intro}.

This exhausts all possible quadrilaterals that could appear in the poset. Hence the relative orientation between any two charts that differ by a quadrilateral is ${-1}$. A significant consequence of this result is the existence of a boundary operator which manifestly squares to zero, as we discuss further in Section \ref{sec:Apps}.

\subsection{The Master Algorithm}
\label{sec:master}
Before proceeding to discuss various applications, we present a more efficient method to compute the relative orientation of any two charts. In each iteration of Algorithm \ref{alg:alg1}, every edge in the reference charts enters into two comparisons (once to the corresponding initial chart, and once to the other reference chart), while the edges in the initial charts enter only one comparison (to the associated reference chart). Therefore, if we assign ${\pm 1}$ to each \textit{edge} such that the product of signs around any quadrilateral is ${-1}$, then the signs on the reference chart edges will appear twice and hence square to 1, while the product of signs on the initial chart edges will combine to produce the same overall sign as found by Algorithm \ref{alg:alg1}. One method for producing a consistent set of edge signs is presented in Appendix \ref{app:edgesigns} \cite{TLDSprivate}.

When changing 0d cells with Algorithm \ref{alg:alg2}, one should think of taking a closed loop that traverses down and back each branch in \eqref{split}, thus encountering each edge twice. Consequently, every edge sign appears twice, thus squaring to 1, so the only sign from this step will be ${-1}$ due to the relative sign between the reference charts. However, one can easily check that applying Algorithm \ref{alg:alg1} will introduce one additional copy of each branch, so those edges should be included in the overall loop. The result of chaining several of these together is a sawtooth path between the two 0d cells (the bold red line in \eqref{sawtooth}) that contributes signs from each edge along the way and a minus sign for each 1d cell along the path. Schematically, the combined loop will look like the following:
\vspace*{-4mm}\als{\includegraphics[height=6cm]{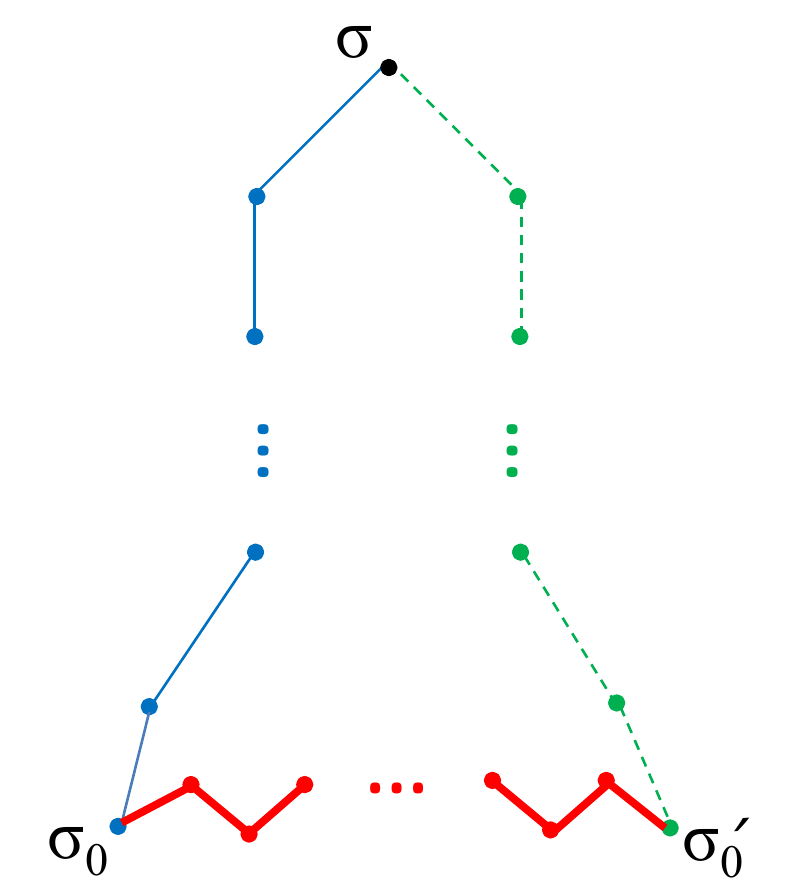}\label{sawtooth}}
\newpage
\noindent This method for computing signs is summarized in the \ref{alg:master} Algorithm:
\begin{alg}[Master Algorithm]
\namedlabel{alg:master}{Master}
\textbf{Input:} Two BCFW sequences of length ${d=k(n-k)}$: ${\w=(a_1 b_1)(a_2 b_2)\ldots (a_d b_d)}$ and
\\\hspace*{2cm} ${\w'=(a_1' b_1')(a_2' b_2')\ldots(a_d' b_d')}$\\
\textbf{Output:} ${\pm 1}$\vspace{2mm}\\
\textbf{Procedure:}\vspace{-5mm}
\begin{enumerate}[label=\bfseries\arabic*)]
\item If both BCFW paths terminate in the same 0-dimensional cell, then the relative orientation is given by the product of edge signs along the paths. Equivalently, it is the product of signs around the closed loop of length ${2d}$ obtained by traversing down one path and back up the other. 
\item If they terminate in different 0d cells, then the relative orientation also depends on the signs along a path connecting the two 0d cells. One can obtain such a path by a sawtooth pattern between 0d and 1d cells. The sign of this path is given by the product of edge signs along the path, times ${(-1)^{m/2}}$, where ${m}$ is the number of locations $i$ satisfying ${\s_0(i)-\s_0'(i)\neq0}$ (${m/2}$ is the number of 1d cells in the sawtooth path). Thus the relative orientation is given by the product of signs along each BCFW path, times the connecting path sign. Equivalently, it is the product of signs around the closed loop obtained by concatenating the three paths \eqref{sawtooth}, times the signs for the 1d cells.
\end{enumerate}
\end{alg}
So far, all the charts have been assumed to have minimal length, i.e.~${k(n-k)}$ for charts on the top cell. In other words, each transposition in the sequence increases the dimension by 1. However, the \ref{alg:master} Algorithm indicates that \emph{any} path through the poset can be compared to any other path, even if they zig-zag up and down. 

We have verified this algorithm by implementing it in Mathematica and applying it to a variety of charts whose orientations are known by other methods. This includes pairs of randomly generated NMHV charts of the type studied in \cite{Elvang:2014fja}, as well as higher ${k}$ charts whose matrix representatives have identical ${GL(k)}$ gauge fixings; the latter can be compared by directly equating the entries. In addition, the orientations of 500 distinct charts on the 10d cell ${\{5, 3, 8, 9, 6, 7, 12, 10, 14, 11\}\in \Gr(3,10)}$ were computed using an independent method due to Jacob Bourjaily and Alexander Postnikov \cite{JBprivate}. The results agreed perfectly with our algorithm.

\section{Applications}
\label{sec:Apps}
In the remainder of this paper, we will discuss several areas in which the relative orientations are important. The end results are not surprising; they were anticipated and used in several previous works. The new contribution of this section will be to put these ideas on firm combinatorial footing, so for example, spurious poles in the tree contour will cancel exactly instead of ${\mod 2}$ as in \cite{ArkaniHamed:2012nw}.

\subsection{Comparing Residue Orientations}
\label{sec:comparing}
We have shown that any two charts can be compared by taking the product of edge signs around a closed loop in the poset of cells.
This applies to any charts, even those corresponding to generalized decompositions with non-adjacent transpositions. As explained in Section \ref{sec:genpaths}, one can follow any path through the poset by taking residues out of order in standard BCFW charts and changing coordinates as needed. Due to the \ref{alg:master} Algorithm, the final sign on the residue depends only on the path, not on the intermediate choices of charts. We will demonstrate this with a convincing example. 

There is a standard chart on ${\{4,3,6,5\}}$ obtained by the path ${P_1}$:
\al{ \{4,3,6,5\} \xrightarrow{(23)} \{4,6,3,5\} \xrightarrow{(14)} \{5,6,3,4\}.}
Both transpositions are adjacent. The corresponding matrix representative and form are
\al{
C=\left(\begin{array}{cccc}
1 & 0 & 0 & \a_1 \\
0 & 1 & \a_2 & 0
\end{array}\right)
\qquad 
\omega = \dlog\a_2 \wedge \dlog \a_1.\label{eg4365}}
Taking either coordinate to vanish lands in a codimension-1 boundary, so we are allowed to take them to zero in either order. There is also a non-adjacent path ${P_2}$ that ends in the same 0d cell:
\al{ \{4,3,6,5\} \xrightarrow{(14)} \{5,3,6,4\} \xrightarrow{(23)} \{5,6,3,4\}.}
We will now evaluate the 0d residue along both paths using the coordinate chart \eqref{eg4365}, ignoring the delta functions in \eqref{integral}. The convention for evaluating residues is to take a contour around ${\a_i=0}$ only when ${\a_i}$ is the first variable in the form. Along ${P_1}$ we first take ${\a_2\to 0}$ and then ${\a_1\to 0}$, which is the order they appear in the form; hence the residue is ${+1}$. We can follow ${P_2}$ by taking ${\a_1\to 0}$ first, so we pick up a factor of ${-1}$ from reversing the order of the wedge product. Thus the residue along ${P_2}$ is ${-1}$. The relative sign is ${-1}$, exactly as our algorithm predicts because the paths differ by a quadrilateral. 

In terms of edge signs, there is a ${\mathbb{Z}_2}$ symmetry at every vertex that allows us to flip the signs on all the attached edges without changing the overall sign of any closed loop. Since the product of signs around a quadrilateral is ${-1}$, we can use the symmetry to fix the edge signs so they exactly agree with the above computation at each step:
\als{\includegraphics[height=4cm]{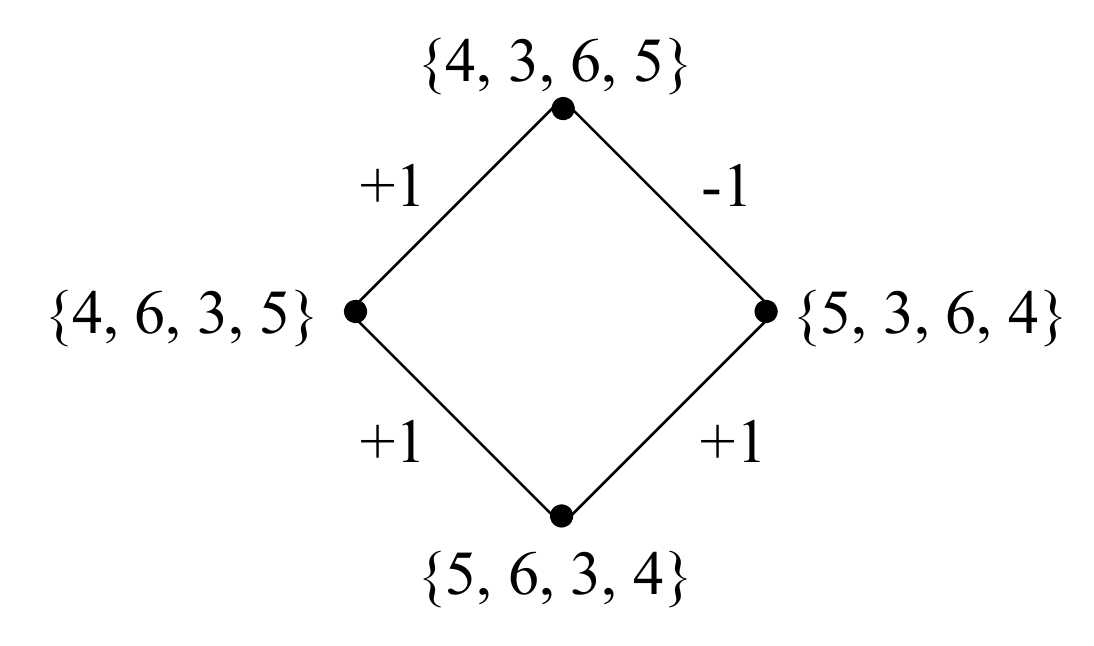}\label{matchingsigns}}
The generalization to more complicated charts is straightforward.

\subsection{Boundary Operator}
Define the signed boundary operator ${\pa}$ acting on a cell ${C}$ to be the sum of all boundaries of ${C}$ weighted by the ${\pm 1}$ weight on the edge connecting each boundary to ${C}$.
\al{\pa C = \sum_i w\big(C,C'_i\big)\, C'_i,\label{pa}}
where the sum is over all cells $C_i'$ in the boundary of ${C}$ and ${w(C,C'_i)=\pm1}$ is the weight on the edge between ${C}$ and ${C'_i}$. Equivalently, we could take the sum over all cells of the appropriate dimension and define ${w(C,C'_i)=0}$ whenever there is no edge between them.%
\footnote{This is not a unique definition since the edge weights can be flipped without affecting signs around closed loops, but any consistent set of signs such as those defined in Appendix \ref{app:edgesigns} will be sufficient for our purposes.} Applying the boundary operator again therefore yields a sum of codimension-2 boundaries of ${C}$, each one weighted by the product of the sign on the edge connecting it to its parent times the sign on its parent from the first application of ${\pa}$.
\al{\pa^2 C = \sum_i \sum_{j(i)} w\big(C,C'_i\big)\,w\big(C'_i,C''_{j(i)}\big) C''_{j(i)},}
where ${i}$ runs over boundaries ${C_i'}$ of ${C}$ and ${j(i)}$ runs over boundaries ${C_ {j(i)}''}$ of ${C'_i}$. In order for this result to vanish, every codimension-2 cell must appear twice and with opposite signs. 

We will first show that each cell appears twice,\footnote{See also Section 6.3 of \cite{ArkaniHamed:2012nw} for a similar proof that they appear twice.} and then it will be clear from our setup that that signs are opposite. Let ${\s}$ be the permutation labeling ${C}$. Each edge represents a transposition acting on ${\s}$, so codimension-2 cells will arise from pairs of transpositions ${(ab)(cd)}$. If ${a,b,c,d}$ are all distinct, then the transpositions can be applied in either order and thus each cell appears twice. If only three are distinct, then there are a few cases to consider:
\als{
\begin{array}{ccl}
(ac)(ab) \equiv (ab)(bc) & & \s(a)<\s(c)<\s(b) \\
(ac)(bc) \equiv (bc)(ab) & & \s(b)<\s(a)<\s(c) \\
(bc)(ac) \equiv (ab)(bc) \text{~and~} (ab)(ac) \equiv (bc)(ab) & & \s(a)<\s(b)<\s(c).
\end{array} \label{equivs}}
Thus there are two unique routes from ${C}$ to every codimension-2 cell in ${\pa^2C}$. Each pair of routes defines a quadrilateral in the poset, which we have seen implies a relative minus sign between the residues.
Hence the boundary operator manifestly squares to zero.

\subsection{Locality of Tree Contours}
The ${n}$-particle N${^k}$MHV tree amplitude can be computed as a linear combination of residues of \eqref{integral} with coefficients ${\pm 1}$.
Each residue appearing in the amplitude corresponds to a ${4k}$-dimensional cell, whose the remaining degrees of freedom are fixed by the ${4k}$ bosonic delta functions in \eqref{integral}.
A \emph{tree contour} is defined as any choice of contour on the top cell that produces a valid representation of the tree amplitude; there are many equivalent representations due to residue theorems. 
Tree-level BCFW recursion relations \cite{Britto:2004ap,Britto:2005fq,Brandhuber:2008pf} written in terms of on-shell diagrams \cite{ArkaniHamed:2012nw} provide one technique to find an appropriate set of cells, but the on-shell diagram formulation does not generate the relative signs between them. This will be resolved shortly.

The tree amplitude diverges for certain configurations of the external momenta; these are called local poles --- physically, they are interpreted as factorization channels in which an internal propagator goes on-shell. In the Grassmannian residue representation, such poles correspond to ${(4k-1)}$-dimensional cells, i.e.~boundaries of the ${4k}$-dimensional cells. The amplitude is not manifestly local in this formulation, meaning that some boundary cells translate to non-local poles, which are momentum configurations with non-physical divergences. A key feature of the tree contour is that all of the local poles appear precisely once in the residue representation of the amplitude, while any non-local poles appear twice \cite{ArkaniHamed:2012nw}. It was conjectured that the two appearances of each non-local pole should come with opposite signs so they cancel in the sum, similar to the vanishing of ${\pa^2}$. We are now equipped to prove that claim.

The boundary operator \eqref{pa} can be used to define signed residue theorems. Even though the sign of each term in \eqref{pa} is not fixed, the \emph{relative} sign between any two cells in the boundary, say ${C'_1}$ and ${C'_2}$, will always agree whenever they both appear in the boundary of a cell. It is easy to see that this is true if ${C'_1}$ and ${C'_2}$ share a common codimension-2 boundary since they form a quadrilateral. The edges connecting ${C'_1}$ and ${C'_2}$ to their shared boundary are the same no matter which ${C}$ is used in the initial boundary operation, so the relative sign between the edges connecting ${C}$ to ${C'_1}$ and ${C'_2}$ must not depend on that choice either. There is one situation in which they may not share a codimension-2 boundary, but in that case there will always be a third cell ${C_3'}$ which has common boundaries with both of them; cf. case (vi) in Section \ref{sec:BCFWgen}. Now following the intuition of \cite{ArkaniHamed:2012nw}, we find residue theorems by requiring that the boundary of every ${(4k+1)}$-dimensional cell vanishes:
\al{\pa C^{(4k+1)}=\sum_i w\Big(C^{(4k+1)},C^{(4k)}_i\Big)\, C^{(4k)}_i=0.\label{resthm}}

We define the tree contour to encircle each singularity of the measure exactly once, so any residue appearing in the amplitude will have a coefficient ${\pm 1}$. The residue theorems \eqref{resthm} can change which poles are included in the contour, but they will never cause residues to appear more than once. 
This implies that the relative sign between any two cells in the amplitude must match the relative sign of those cells in the residue theorems. Therefore, by the same logic that showed ${\pa^2=0}$, it follows that any ${(4k-1)}$-dimensional cell appearing twice in the boundary of the tree amplitude will show up with opposite signs. Hence all non-local poles cancel in the sum.

We have checked numerically that this choice of signs correctly cancels all non-local poles for BCFW representations of the tree amplitude with ${n=5,\ldots,13}$ and ${k=3,\ldots,\lfloor n/2 \rfloor}$.

\section*{Acknowledgments}
A great deal of gratitude is owed to Thomas Lam and Rachel Karpman for their patience and helpful discussions throughout the course of this project, and I am grateful to Jake Bourjaily for sharing his Mathematica code and for related discussions. I would like to thank Henriette Elvang, Yu-tin Huang, Cindy Keeler, Samuel Roland, and David Speyer for many helpful conversations along the way. Thanks also to Emily Olson for providing feedback and comments on the paper. This material is based upon work supported by the National Science Foundation Graduate Research Fellowship under Grant No.~F031543.

\appendix

\section{Details for Algorithm \ref{alg:alg1}}
\label{app:alg1}
In this appendix, we derive the reference sequences and compute the relative signs for step (2) of Algorithm \ref{alg:alg1}.
\begin{enumerate}[label=\bfseries\roman*)]
\item \textbf{${\bm{(ab)}}$ vs. ${\bm{(cd)}}$}\\
This case was discussed in the main text so we do not repeat the argument here.
\item \textbf{${\bm{(ab)}}$ vs. ${\bm{(ac)}}$}\\
These transpositions share a leg, so they do not commute as simply as in the previous case. Nonetheless, in this case we can use the fact that both ${(ab)}$ and ${(ac)}$ are allowed transpositions on the initial permutation to find a common parent cell. Specifically, since ${(ab)}$ was the last transposition before ${\s}$, we know ${a<b\leq\s(a)<\s(b)\leq a\!+\!n}$ and there is no ${q\in(a,b)}$ such that ${\s(q)\in (\s(a),\s(b))}$. Similarly from ${(ac)}$ we know ${a<c\leq\s'(a)<\s'(c)\leq a\!+\!n}$ and there is no ${q\in(a,c)}$ such that ${\s'(q)\in (\s'(a),\s'(c))}$. Moreover, since ${\s}$ and ${\s'}$ come from the same initial permutation, we also know that ${\s(a)=\s'(b)}$, ${\s(b)=\s'(c)}$, ${\s(c)=\s'(a)}$, and ${\s(q)=\s'(q)}$ for all other legs.

To show that ${(bc)}$ can be applied to ${\s}$, we need to show that ${\s(c) < \s(b)}$ and that there is no ${q\in (b,c)}$ such that ${\s(q)\in (\s(c),\s(b))}$. The condition that ${\s(c)<\s(b)}$ is satisfied since it is equivalent to ${\s'(a)<\s'(c)}$, and since no legs are touched between ${b}$ and ${c}$, the condition on ${q\in(a,c)}$ implies that no ${q\in(b,c)}$ has ${\s(q)\in (\s(c),\s(b))}$. Thus ${(bc)}$ is an allowed transposition on ${\s}$, so the path ${\tw}$ exists.

We also need to show that ${(ab)}$ can be applied to ${\s'}$, so we must show ${\s'(b) < \s'(a)}$ and that there is no ${q\in (a,b)}$ such that ${\s'(q)\in (\s'(b),\s'(a))}$. Since ${b\in(a,c)}$, we must have either ${\s'(b)<\s'(a)<\s'(c)}$ or ${\s'(a)<\s'(c)<\s'(b)}$. Since ${\s(a)<\s(b) \Rightarrow \s'(b)<\s'(c)}$, only the former condition is allowed, hence ${\s'(b)<\s'(a)}$. Furthermore, there is no ${q\in(a,b)}$ such that ${\s'(q)\in(\s'(b),\s'(c))}$, and ${(\s'(b),\s'(a))}$ is a subset of that range. Hence ${(ab)}$ is allowed on ${\s'}$, and ${\tw'}$ exists.

Note that the above analysis was valid for any charts, not just adjacent ones. However, to compute the relative sign, we will focus on the restricted set of charts for which all ${q\in(a,c)}$ satisfy ${\s'(q)\equiv q ~\mod n}$, namely the standard BCFW charts. The corresponding plabic graphs can be manipulated to a common layout using the equivalence moves \eqref{e1} and \eqref{e2}:
\al{\spl{\includegraphics[height=3.5cm]{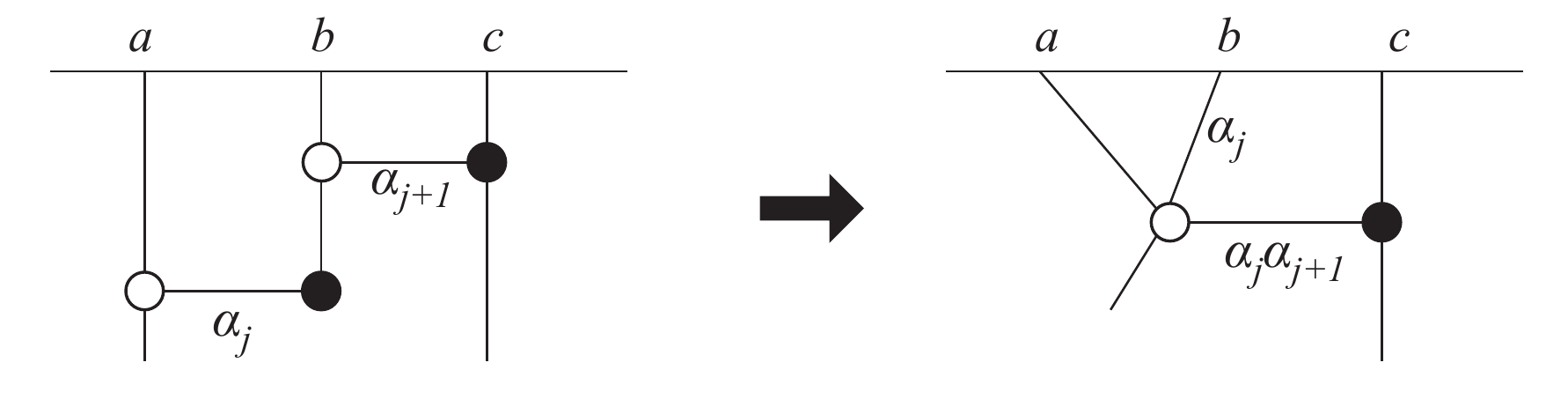}}\\
\spl{\includegraphics[height=3.5cm]{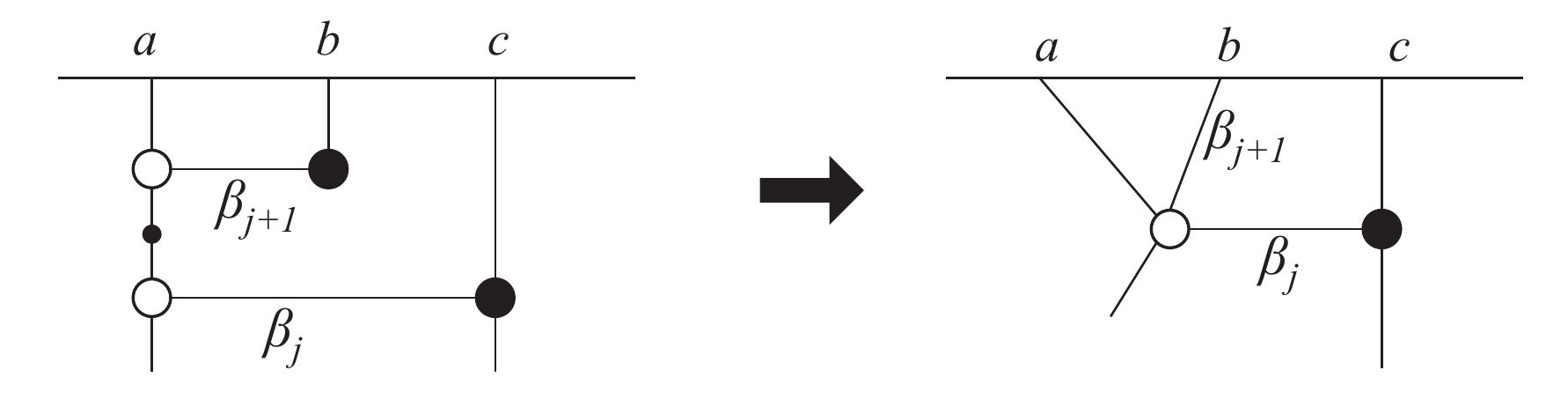}}}
We find that they are equivalent only with the identifications 
\als{\b_j=\a_j\a_{j+1}, \qquad \b_{j+1} = \a_j.}
Plugging this into the dlog forms and using that 
\als{
\dlog\a_j \wedge \dlog (\a_j\a_{j+1}) = \dlog\a_j \wedge \Big(\dlog\a_j + \dlog \a_{j+1}\Big) = -\dlog\a_{j+1} \wedge\dlog \a_{j},
}
we find that the two forms are oppositely oriented.
\item \textbf{${\bm{(ac)}}$ vs. ${\bm{(bc)}}$}\\
This situation is analogous to case (ii), so we can skip directly to comparing the graphs. Restricting again to the standard BCFW situation where ${\s(q)\equiv q ~\mod n}$ for all ${q\in(a,c)}$, we can perform a sequence of merge/delete and ${GL(1)}$ rotations on the corresponding plabic graphs:
\al{\spl{\includegraphics[height=3.5cm]{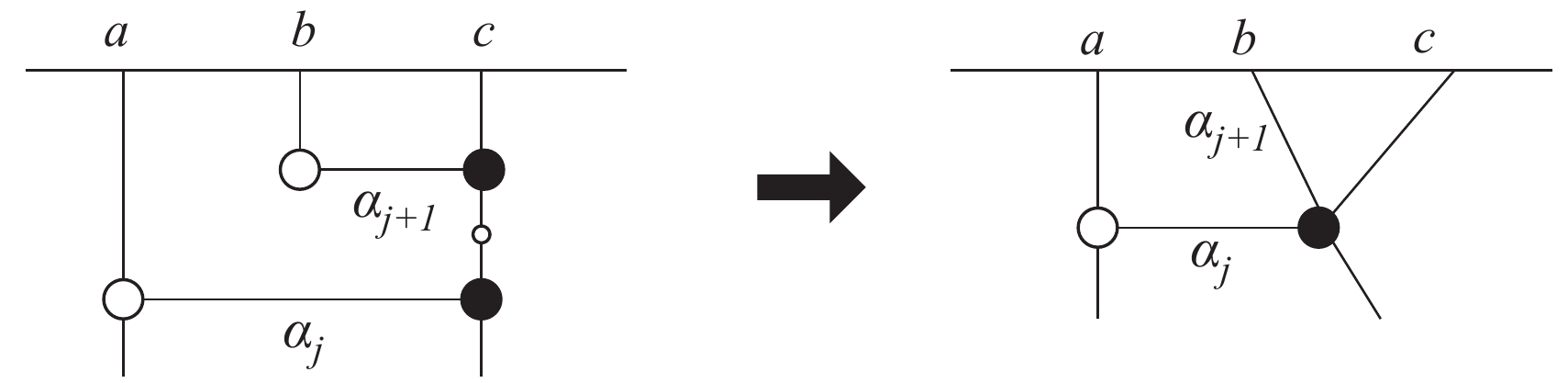}}\\
\spl{\includegraphics[height=3.5cm]{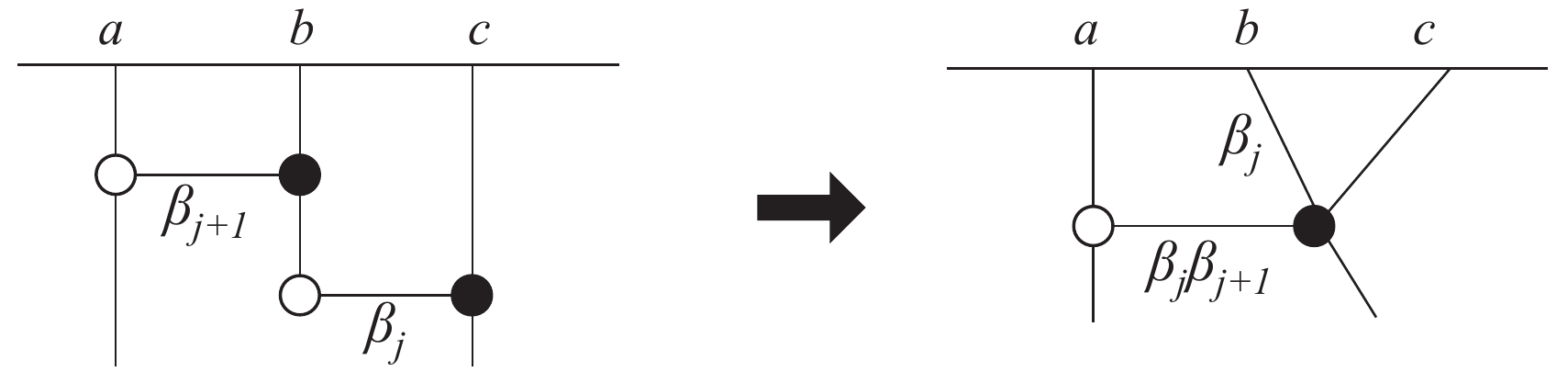}}}
We therefore identify 
\als{\b_j=\a_{j+1}, \qquad \b_{j+1}=\a_j/\a_{j+1}.}
Then using that ${\dlog(\a_j/\a_{j+1})=\dlog\a_j - \dlog\a_{j+1}}$, we find that the two forms are oppositely oriented.
\item \textbf{${\bm{(ab)}}$ vs ${\bm{(bc)}}$}\\
Using only adjacent transpositions, ${\s}$ and ${\s'}$ do not have a common parent cell. However, we can show that ${\rho=\s\cdot(bc)(ab)=\s'\cdot (ab)(bc)}$ is a shared grandparent. From the initial ${(ab)}$, we know ${a<b\leq \s(a) <\s(b) \leq a\!+\!n}$, while from  the initial ${(bc)}$, we have ${b<c\leq \s'(b) < \s'(c) \leq b\!+\!n}$.
From their shared origin cell, we also have ${\s(a)=\s'(c)}$, ${\s(b)=\s'(a)}$, and ${\s(c)=\s'(b)}$.
We will focus on adjacent charts wherein ${\s(q)=\s'(q) \equiv q ~\mod n}$ for all ${q\in (a,b)\cup (b,c)}$. The general case is covered in Section \ref{sec:BCFWgen}.

Since ${\s'(b)<\s'(c)}$ is equivalent to ${\s(c)<\s(a)}$, and ${\s(a)<\s(b)}$, we can apply ${(bc)}$ to reach ${\ts=\s\cdot (bc)}$. Now ${\ts(b) < \ts(a) < \ts(c)}$, so ${(ab)}$ is a valid transposition, which arrives at ${\rho}$. Hence the reference chart ${\tw}$ exists.

On the other side, ${\s(a)<\s(b)}$ is equivalent to ${\s'(c)<\s'(a)}$, and ${\s'(b)<\s'(c)}$, so ${(ab)}$ can be applied to ${\s'}$, yielding ${\ts'=\s'\cdot(ab)}$. Since ${\ts'(a)<\ts'(c)<\ts'(b)}$, we can apply ${(bc)}$, which also arrives at ${\rho}$. Thus ${\tw'}$ is also a valid reference chart.

Finally, we can compare the plabic graphs to find their relative orientation:
\al{\spl{\includegraphics[height=4cm]{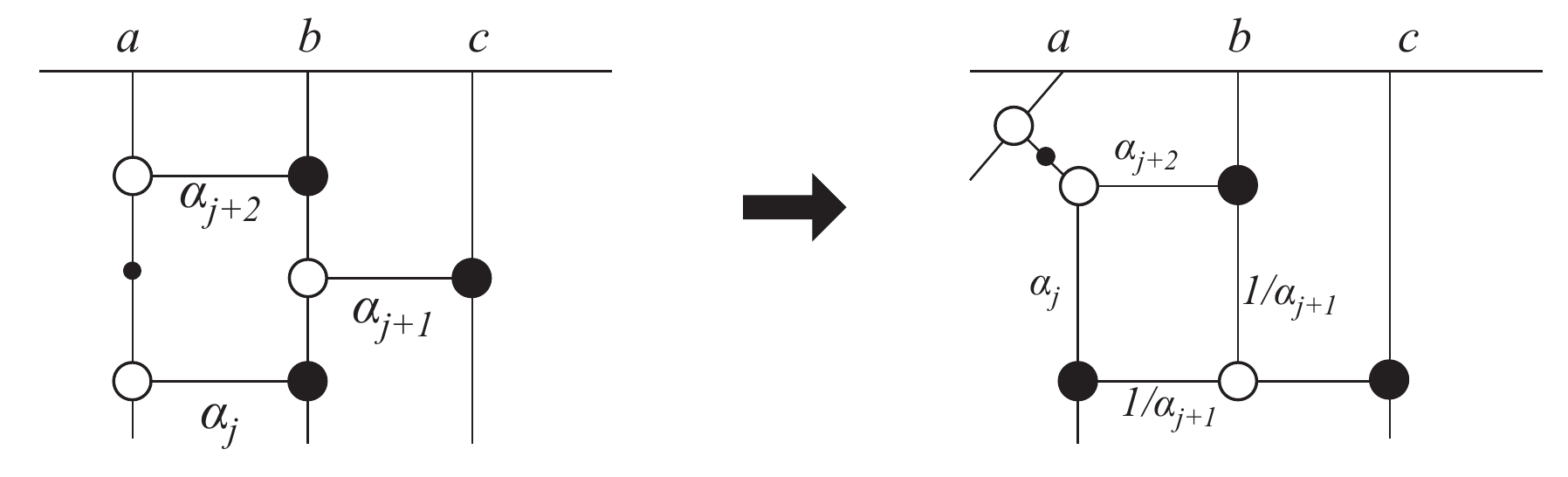}\label{braid1}}\vspace{-1cm}\\
\spl{\includegraphics[height=4cm]{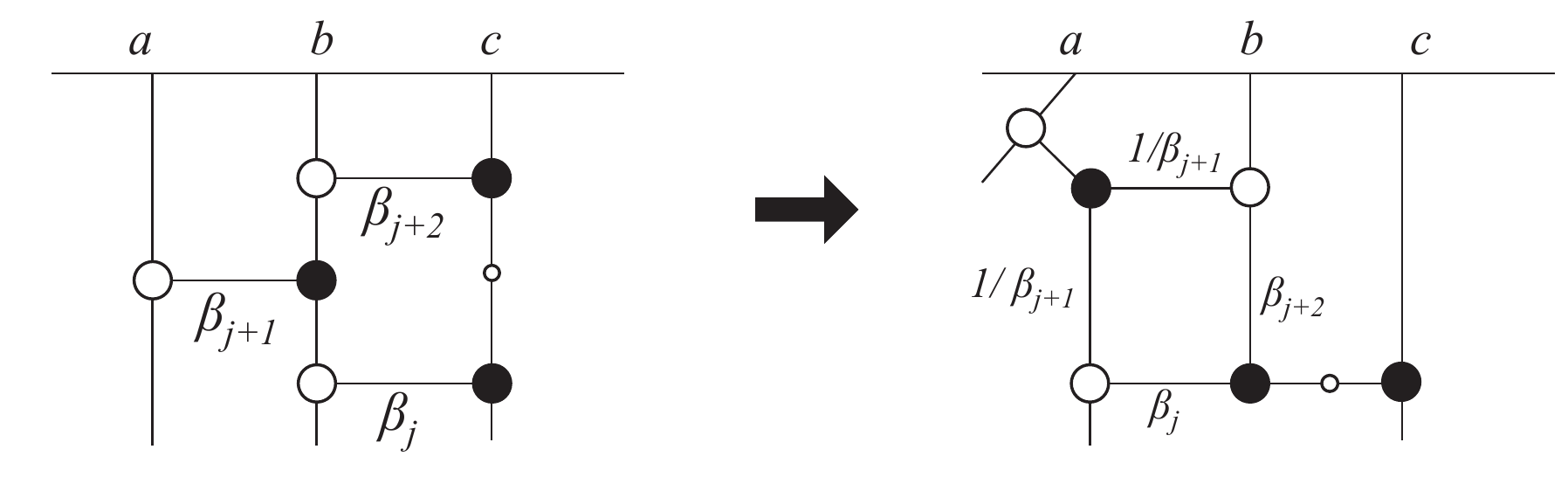}\label{bcabbc}}}
After performing a square move \eqref{e3} on \eqref{bcabbc}, we find
\als{\beta_j=\frac{\a_{j+1}\a_{j+2}}{\a_j+\a_{j+2}} , \quad \beta_{j+1}={\a_j+\a_{j+2}},\quad  \beta_{j+2}=\frac{\a_j\a_{j+1}}{\a_j+\a_{j+2}}.}

Chugging through a bit of algebra, the result is a positive relative orientation.
\item \textbf{${\bm{(ac)}}$ vs. ${\bm{(bd)}}$}\\
There is no way to find a common parent using only adjacent transpositions. However, there is a common grandparent ${\rho=\s\cdot (bc)(cd) =\s'\cdot(bc)(ab)}$. Since ${\s}$ and ${\s'}$ come from a shared origin, we have the following: ${a<c\leq \s(a) <\s(c) \leq a\!+\!n}$ and ${b<d\leq \s'(b) <\s'(d) \leq b\!+\!n}$; and ${\s(a)=\s'(c)}$, ${\s(b)=\s'(d)}$, ${\s(c)=\s'(a)}$, and ${\s(d)=\s'(b)}$. 
We focus on adjacent charts, so ${\s(q)=\s'(q)\equiv q ~\mod n}$ for all ${q\in (a,d)\backslash \{b,c\}}$. In addition, ${\s(b)=\s'(d) = b\!+\!n}$ and ${\s'(c)=\s(a)=c}$.

Since ${\s(b)=b\!+\!n>a\!+\!n\geq \s(c)}$, the transposition ${(bc)}$ is allowed, which leads to ${\ts=\s\cdot (bc)}$. Then ${\ts(a)<\ts(b)<\ts(c)=b\!+\!n}$ and ${\ts(d)=\s(d)=\s'(b)<b\!+\!n}$. Therefore ${(cd)}$ is also allowed, and we arrive at ${\rho}$. Therefore ${\tw}$ is a valid reference chart.

For ${\s'}$, we use that ${\s'(c)=c < d \leq \s'(b)}$, so ${(bc)}$ is allowed. Thus with ${\ts'=\s'\cdot (bc)}$, we have that ${c=\ts'(b)<\ts'(c)<ts'(d)}$ and ${\ts'(a)=\s'(a)=\s(c)>c}$. Hence we can apply ${(ab)}$, which yields ${\rho}$, so ${\tw'}$ is a good reference chart.

To compute the relative sign, we study the corresponding plabic graphs:
\al{\spl{\includegraphics[height=4cm]{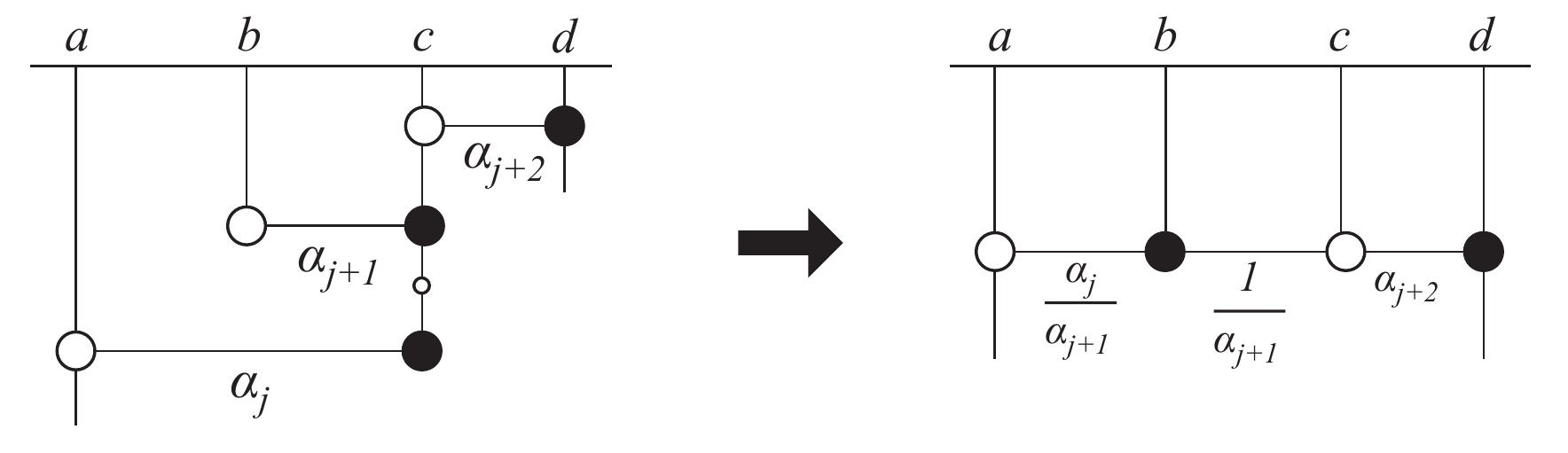}}\\
\vspace{-1cm}
\spl{\includegraphics[height=4cm]{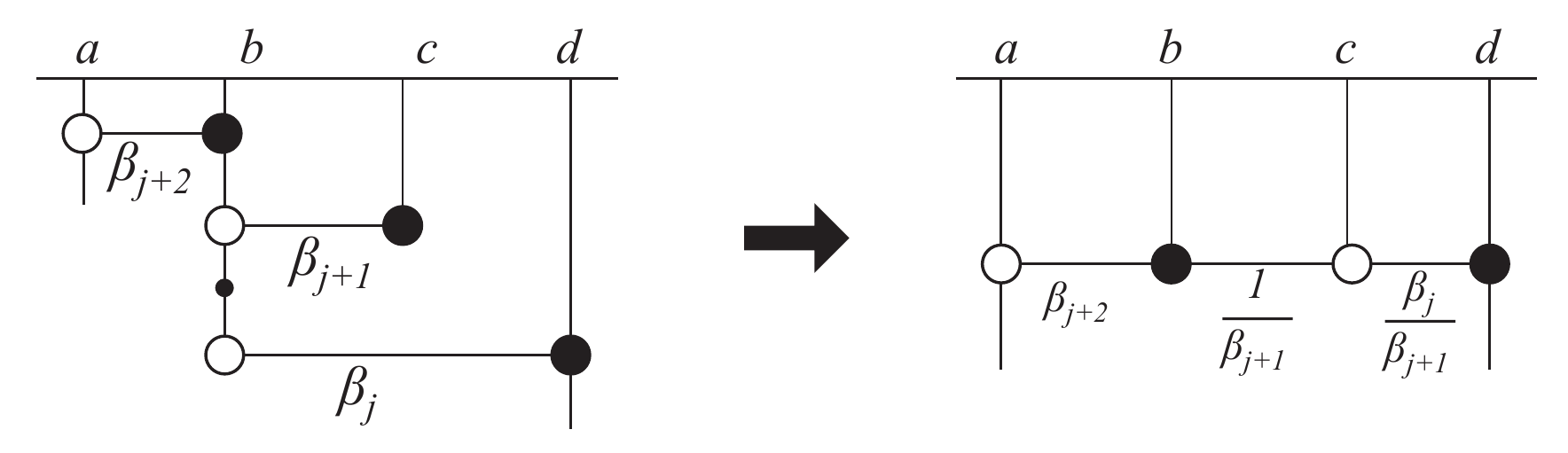}}}
They are equivalent under the identifications
\als{\beta_j=\a_{j+1}\a_{j+2}, \quad \beta_{j+1}=\a_{j+1}, \quad \beta_{j+2}={\a_j}/{\a_{j+1}}.}
Plugging this into ${\omega'}$, one finds that the two forms are oppositely oriented.
\item \textbf{${\bm{(bc)}}$ vs. ${\bm{(ad)}}$}\\
From the point of view of a graph embedded in a disk, these bridges can be added in either order:
\als{\includegraphics[height=4cm]{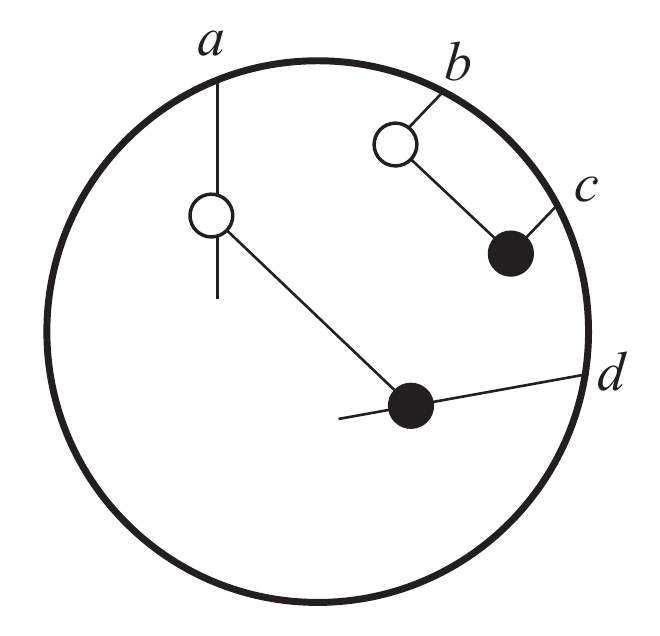}}
However, the adjacency requirement forbids applying ${(ad)}$ after ${(bc)}$. Thus we must look further to find a meeting point, ${\tr}$, for their reference charts. We will focus on adjacent charts, leaving the general case for Section \ref{sec:BCFWgen}.  

Since ${b,c\in(a,d)}$, we know ${\s(b),\s(c)\equiv c,b~ \mod n}$, so it follows that ${\s(b)=c}$ and ${\s(c)=b\!+\!n}$. Then since ${\s(a)=\s'(d)\geq d > c =\s(b)}$, we can next apply ${(ab)}$.
Similarly, ${\s(d)=\s'(a)\leq a+n < b+n = \s(c)}$, we could also apply ${(cd)}$; this is unaffected by applying ${(ab)}$, so we apply it next to arrive at ${\r=\s\cdot(ab)(cd)}$. Finally, $(bc)$ is allowed because ${\r(b)=\s(a)=\s'(d)>\s'(a)=\s(d)=\r(c)}$ and ${c<d\leq\s'(a)<\s'(d)\leq a\!+\!n < b\!+\!n}$. Hence ${\tw}$ is a valid reference chart.

From ${\s'}$, we are certainly allowed to apply ${(bc)}$ after ${(ad)}$, thus arriving at ${\ts=\s\cdot(bc)}$. We can add ${(ab)}$ and then ${(cd)}$ by essentially the same argument as above. Therefore both ${\tw}$ and ${\tw'}$ are valid reference charts.

We compare the plabic graphs to find the relative sign. One can either use the equivalence moves \eqref{e1}-\eqref{e3} or repeatedly apply the transformation rules from cases (ii)-(iv) to find
\al{\spl{\includegraphics[height=4cm]{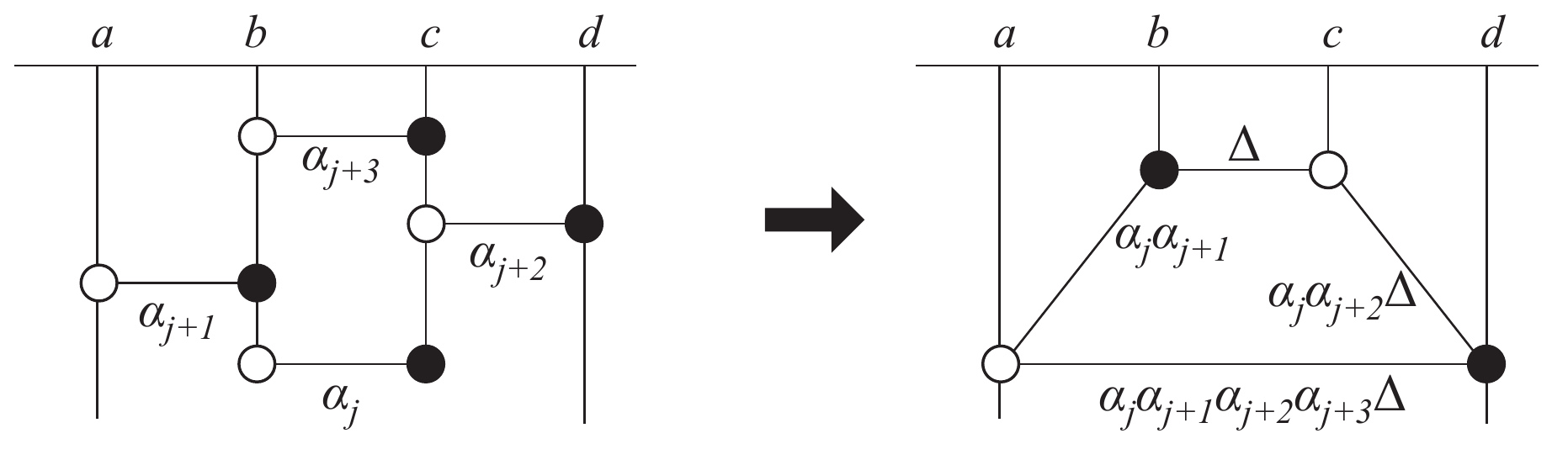}}\\
\spl{\includegraphics[height=4cm]{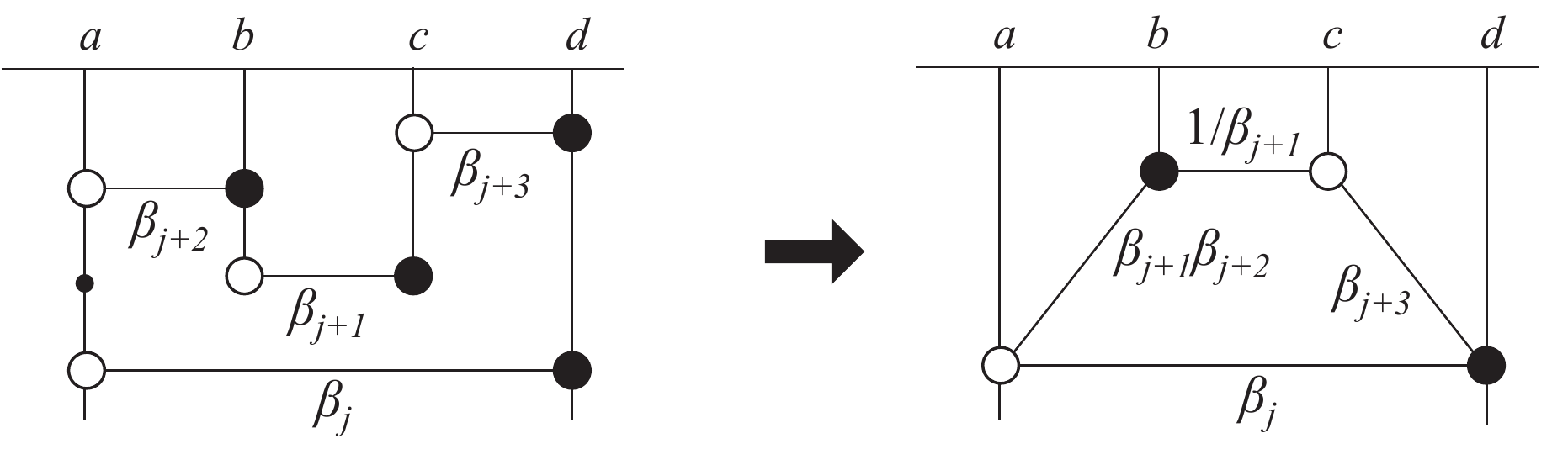}}}
where we have defined
$\Delta=1\big/(\a_j+\a_{j+3})$.
The graphs are equivalent after the following identifications:
\als{\b_j =\frac{\a_j \a_{j+1} \a_{j+2} \a_{j+3}}{\a_j+\a_{j+3}} , ~~ \b_{j+1}=\a_j+\a_{j+3}, ~~ \b_{j+2} =\frac{\a_j \a_{j+1}}{\a_j+\a_{j+3}}, ~~ \b_{j+3}=\frac{\a_j \a_{j+2}}{\a_j+\a_{j+3}}.}
Rearranging the corresponding forms demonstrates that the relative orientation of the reference charts is ${-1}$.
\end{enumerate}

\section{Assigning Edge Weights}
\label{app:edgesigns}
In this appendix we provide one technique for assigning weights ${\pm1}$ to every edge in the poset such that the product of signs around every quadrilateral is ${-1}$. The method was developed by Thomas Lam and David Speyer \cite{TLDSprivate}.

\subsection{Reduced Words}
Each decorated permutation can equivalently be represented by one or more \emph{reduced words}: minimal-length sequences of \emph{letters} which obey certain equivalence relations. In this case, the letters ${s_i}$ are generators of the affine permutation group ${\tilde{S}_n}$. They are defined with the following properties and relations:
\beq
\spl{
&s_i := \left\{\begin{array}{c}\sigma(i)\to\sigma(i+1)\\ \sigma(i+1)\to\sigma(i) \end{array}\right. ~ 1\leq i\leq n,  \qquad s_{i\!+\!n}\equiv s_i,  \\
&s_i^2 \equiv 1, \qquad s_is_{i+1}s_i\equiv s_{i+1}s_is_{i+1}, \qquad s_i s_j \equiv s_j s_i\, ,~ |i-j|>1.
}
\label{reducedWords}
\eeq
We will refer to the relations in the second line as, respectively, the \textit{reduction move}, the \textit{braid move}, and the \textit{swap move}.
As we showed in Lemma \ref{lemma}, any ${d}$-dimensional cell ${C}$ in ${\Gr(k,n)}$ can be reached from the top cell by a sequence of ${k(n-k)-d}$ transpositions ${s_i}$, and that sequence defines a reduced word labeling ${C}$. There are generally many distinct reduced words for a given cell, but they are equivalent due to the relations in \eqref{reducedWords}. 

Any two cells ${C}$ and ${C'}$ of dimensions ${d}$ and ${d+1}$, which share an edge in the poset, i.e.~are related by a single boundary operation, are straightforwardly connected in the language of reduced words. Given a reduced word ${f}$ of length ${\delta=k(n-k)-d}$ on the lower-dimensional cell ${C}$, there is a unique reduced word ${f'}$ on ${C'}$ of length ${\delta'=\delta-1}$ obtained by deleting a single letter from ${f}$ \cite{bjorner2005combinatorics}. Specifically, for ${f=s_{i_1}s_{i_2}\ldots s_{i_j} \ldots s_{i_\delta}}$, there is a unique ${s_{i_j}}$ such that ${f'=s_{i_1}s_{i_2}\ldots\hat{s}_{i_j}\ldots s_{i_\delta}}$, where the hat denotes deletion. It is easy to construct a (generally non-reduced) word ${t}$ such that ${f'\equiv tf}$; one can check that ${t=s_{i_1}s_{i_2}\ldots s_{i_{j-1}} s_{i_j}s_{i_{j-1}}\ldots s_{i_2}s_{i_1}}$ accomplishes the desired effect by repeated reduction moves. Given that the two cells are also related by a transposition of the form ${(ab)}$, it is not surprising that repeated application of the relations \eqref{reducedWords} shows that ${t\equiv s_a s_{a+1}\ldots s_{b-2} s_{b-1} s_{b-2} \ldots s_{a+1} s_a}$, which is a reduced word representation of ${(ab)}$.

\subsection{Decorating Edges}
Choose a representative reduced word for every cell in the poset; we will refer to this as the \emph{standard word} on that cell. This is similar to choosing a particular Lemma \ref{lemma} sequence for each cell. From a cell ${C}$ labeled by the standard word ${f}$, every cell ${C'}$ that has ${C}$ as a boundary can be reached by deleting some ${s_{i_j}}$ from ${f}$. This yields a word ${f'}$ on ${C'}$ as explained above. If ${f'}$ is identical to the standard word on ${C'}$, then weight the corresponding edge with ${(-1)^{j}}$. It is easy to check that deleting two letters ${s_{i_{j_1}}}$ and ${s_{i_{j_2}}}$ in opposite orders will produce the desired factor of ${-1}$ around a quadrilateral. For one order, we would find 
${(-1)^{i_{j_1}+i_{j_2}}}$, while for the other order we would find ${(-1)^{i_{j_1}+i_{j_2}-1}}$, so they differ by ${-1}$.

However, it is not always possible to delete both transpositions in either order. One may have to delete different transpositions to reach the same cell by two different routes. Moreover, the final reduced word in each case may be different. Thus, we need to find the sign difference between two reduced words. Any two reduced words can be mutated into each other using just the braid and swap moves defined in \eqref{reducedWords}. By decorating these rules with ${\pm1}$, we can find the desired difference between the words. In fact, we already determined those signs in Appendix \ref{app:alg1} because the generators act as adjacent transpositions. The braid move is simply a special case of (iv), so it should be decorated with a ${+1}$, and the swap move is a special case of (i), so it should be decorated with a ${-1}$. Hence, the relative sign between two words is the number of swaps required to transform one into the other.

The number of swaps can be determined directly from the \emph{inversion list} of each word. A reduced word creates a unique ordering on the set of inversions in the permutation labeling the cell (not all orderings are possible). The number of swap moves needed to transform one word into another is the number of pairs of inversions ${(i,j)}$ and ${(k,l)}$, with all ${i,j,k,l}$ distinct, in different order in the two inversion lists.

\end{document}